\documentclass[12pt]{article}
\usepackage{amsmath}
\usepackage{graphicx,psfrag,epsf}
\usepackage{enumerate}
\usepackage{natbib}
\usepackage{url} % not crucial - just used below for the URL 

\newcommand{\blind}{0}

\RequirePackage[OT1]{fontenc}
\RequirePackage{amsthm}
\usepackage[font=scriptsize]{caption}
\usepackage{soul}
\usepackage[normalem]{ulem}
\newcommand{\stkout}[1]{\ifmmode\text{\sout{\ensuremath{#1}}}\else\sout{#1}\fi}
\usepackage{comment}
\usepackage[textsize=tiny]{todonotes}

\usepackage{adjustbox}
\usepackage{amsfonts}
\usepackage{amssymb}
\usepackage[makeroom]{cancel}
\usepackage{bbold}
\usepackage{float}

\usepackage[colorlinks=true,
            linkcolor=blue,
            urlcolor=blue,
            citecolor=blue]{hyperref}

\usepackage{algorithm}% http://ctan.org/pkg/algorithms
\usepackage{algpseudocode}% http://ctan.org/pkg/algorithmicx
\usepackage{color}

\usepackage{etoolbox}
\usepackage{csvsimple}

\newcommand{\B}[1]{\mathbf{#1}}
\newcommand{\M}{\mathfrak{m}}
\newcommand{\MS}{\cal M}

\newcommand{\BB}[1]{\boldsymbol{#1}}
\AtBeginEnvironment{tabular}{\small}

\newcommand{\By}{\boldsymbol{\mathsf{y}}}

\newtheorem{theorem}{Theorem}

\newtheorem{proposition}{Proposition}

\newtheorem{constraint}{Constraint}
\definecolor{Florian}{RGB}{255,0,0}

\definecolor{Geir}{RGB}{0,0,255}

\definecolor{AH}{RGB}{255,0,255}

\begin{document}

\def\spacingset#1{\renewcommand{\baselinestretch}%
{#1}\small\normalsize} \spacingset{1}

\if0\blind
{
  \title{\bf Deep Bayesian regression models}
 \author{Aliaksandr Hubin \thanks{
    The authors gratefully acknowledge \textit{CELS project at the University of Oslo} for giving us the opportunity, inspiration and motivation to write this article.}\hspace{.2cm}\\
    Geir Storvik \\
    Department of Mathematics, University of Oslo\\
     and \\
    Florian Frommlet \\
    CEMSIIS, Medical University of Vienna}
  \maketitle
} \fi

\if1\blind
{
  \bigskip
  \bigskip
  \bigskip
  \begin{center}
    {\LARGE\bf Deep non-linear regression models in a Bayesian framework}
\end{center}
  \medskip
} \fi

\bigskip
\begin{abstract}
\normalsize
Regression models are used for inference and prediction in a wide range of applications providing a powerful scientific tool for researchers and analysts from different fields. In many research fields the amount of available data as well as the number of potential explanatory variables is rapidly increasing. Variable selection and model averaging have become extremely important tools for improving inference and prediction. However, often  linear models are not sufficient and the complex relationship between input variables and a response is better described by introducing non-linearities and complex functional interactions. Deep learning models have been extremely successful in terms of prediction although they are often difficult to specify and potentially suffer from overfitting. The aim of this paper is to bring the ideas of deep learning into a statistical framework which yields more parsimonious models and allows to quantify model uncertainty. To this end we introduce the class of deep Bayesian regression models (DBRM) consisting of a generalized linear model combined with a comprehensive non-linear feature space, where non-linear features are generated just like in deep learning but combined with variable selection in order to include only important features. DBRM can easily be extended to include latent Gaussian variables to model complex correlation structures between observations, which seems to be not easily possible with existing deep learning approaches. Two different  algorithms based on MCMC are introduced to fit DBRM and to perform Bayesian inference. The predictive performance of these algorithms is compared with a large number of state of the art algorithms. Furthermore we illustrate how DBRM can be used for model inference in various applications. 

\end{abstract}

\noindent%
{\it Keywords:}   Bayesian deep learning; Deep feature engineering and selection; Combinatorial optimization; Uncertainty in deep learning; Bayesian model averaging; Semi-parametric statistics; Automatic neural network configuration; Genetic algorithm;  Markov chains Monte Carlo.
%\vfill

%\newpage
%\spacingset{1.45} % DON'T change the spacing!
\spacingset{1.2} 

\section{Introduction}\label{section1}

Regression models are an indispensible tool for answering scientific questions in almost all research areas. Traditionally scientists have been trained to be extremely careful in specifying adequate models and to include not too many explanatory variables. The orthodox statistical approach warns against blindly collecting data of too many variables and relying on automatic procedures to detect the important ones~\citep[see for example][]{BA_2002}. Instead, expert based knowledge of the field should guide the model building process such that only a moderate number of models are compared to answer specific research questions.
 
In contrast, modern technologies have lead to the entirely different paradigm of machine learning where routinely  extremely large sets of input explanatory variables - the so called features - are considered. Recently deep learning procedures have become quite popular and highly successful in a variety of real world applications \citep{Goodfellow-et-al-2016}. These algorithms apply iteratively some nonlinear transformations aiming at optimal prediction of response variables from the outer layer features. Each transformation yields another hidden layer of features which are also called neurons. The architecture of a deep neural network then includes the specification of the nonlinear intra-layer transformations (\emph{activation functions}), the number of  layers (\emph{depth}), the number of features at each  layer (\emph{width}) and the connections between the neurons (\emph{weights}). The resulting model is trained by means of some optimization procedure (e.g. stochastic gradient search) with respect to its parameters in order to fit a particular objective (like minimization of RMSE, or maximization of the likelihood, etc.).

Surprisingly it is often the case that such procedures easily outperform traditional statistical models, even when these were carefully designed and reflect expert knowledge~\citep{refenes1994stock, razi2005comparative, adya1998ective, sargent2001comparison, kanter2015deep}. Apparently the main reason for this is that the features from the outer layer of the deep neural networks become highly predictive after being processed through the numerous optimized nonlinear transformations. Specific regularization techniques (dropout, $L_1$ and $L_2$ penalties on the weights, etc.) have been developed for deep learning procedures to avoid overfitting of training data sets, however success of the latter is not obvious. Normally one has to use huge data sets to be able to produce generalizable neural networks. 

The universal approximation theorems~\citep{Cybenko1989,hurnik1991} prove that all neural networks with sigmoidal activation functions~\citep[generalized to the class of monotonous bounded functions in ][]{hurnik1991} with at least one hidden layer can approximate any function of interest defined on a closed domain in the Euclidian space.
Successful applications typically involve huge datasets where even nonparametric methods can be efficiently applied. One drawback of deep learning procedures is that, due to their complex nature, such models and their resulting parameters are difficult to interpret. Depending on the context this can be a more or less severe limitation. These models are densely approximating the function of interest and transparency is not a goal in the traditional applications of deep learning. However, in many research areas it might be desirable to obtain interpretable (nonlinear) regression models rather than just some dense representation of them. Another problem is that fitting deep neural networks is very challenging due to the huge number of parameters involved and the non-concavity of the likelihood function. As a consequence optimization procedures often yield only local optima as parameter estimates.

This paper introduces a novel approach which combines the key ideas of deep neural networks with Bayesian regression resulting in a flexible and broad framework that we call deep Bayesian regression. This framework also includes many other popular statistical learning approaches.  Compared to deep neural networks we do not have to prespecify the architecture but our deep regression model can adaptively learn the number of layers, the number of features within each layer and the activation functions. In a Bayesian model based approach potential overfitting is avoided through appropriate priors which strongly penalize engineered features from deeper layers. Furthermore deep Bayesian regression allows to incorporate correlation structures via latent Gaussian variables, which seems to be rather difficult to achieve within traditional deep neural networks.

Fitting of the deep Bayesian regression model is based on a Markov chain Monte Carlo (MCMC) algorithm for Bayesian model selection which is embedded in a genetic algorithm for feature engineering. A similar algorithm was previously introduced in the context of logic regression~\citep{Hubin2017}. We further develop a reversible version of this genetic algorithm to obtain a proper Metropolis-Hastings algorithm.  
We will demonstrate that automatic feature engineering within regression models combined with Bayesian variable selection and model averaging can improve predictive abilities of statistical models whilst keeping them reasonably simple, interpretable and transparent.  The predictive ability of deep Bayesian regression is compared  with deep neural networks, CARTs, elastic networks, random forests, and other statistical learning techniques under various scenarios. Furthermore we illustrate the potential of our approach to find meaningful non-linear models and infer on parameters of interest. As an example we will retrieve several ground physical laws from raw data.

The rest of the paper is organized as follows. The class of deep Bayesian regression models (DBRM) is  mathematically defined in Section~\ref{section2}. In Section~\ref{section3} we describe two algorithms for fitting DBRM, namely the genetically modified MJMCMC (GMJMCMC) and its reversible version (RGMJMCMC). In Section~\ref{section4} these algorithms are applied to several real data sets. The first examples are aiming at prediction where the performance of our approach is compared with various competing statistical learning algorithms. Later examples have the specific goal of retrieving an interpretable model. In the final Section \ref{section5} some conclusions and suggestions for further research are given. Additional examples and details about the implementation can be found in the Appendix.

\section{{DBRM: The deep Bayesian regression model}}\label{section2}
We model the relationship between $m$ features and a response variable based on $n$ samples from a training data set. Let $Y_{i}$ denote the response data and $\boldsymbol{x}_{i} = (x_{i1}, \dots, x_{im})$  the $m$-dimensional vector of input covariates for $i  = 1,...,n$. The proposed model is within the framework of a generalized linear model,
but extended to include a flexible class of non-linear transformations (features) of the covariates, to be further described in Section~\ref{sec:feature}. This class includes a finite (though huge) number $q$ of possible features, which can in principle be enumerated as $F_j(\boldsymbol{x}_i)$, $j = 1,...,q$. With this notation we define the deep Bayesian regression model (DBRM), including (potentially) up to $q$ features: 
\begin{subequations}\label{themodeleq}
\begin{align} 
  Y_i|\mu_i \sim & \ \mathfrak{f}(y|\mu_i;\phi)\; ,&& i  = 1,...,n\\
   \mathsf{h}(\mu_i) = &\ \beta_0  + \sum_{j=1}^{q} \gamma_{j}\beta_{j}F_{j}(\boldsymbol{x}_{i})\; .
   \label{DeepModel}
\end{align}
\end{subequations}
Here $\mathfrak{f}(\cdot|\mu,\phi)$ is the density (mass) of a probability distribution from the exponential family with expectation $\mu$ and dispersion parameter $\phi$, while $\mathsf{h}(\cdot)$ is a link function relating the mean to the underlying covariates \citep{McCullagh-Nelder-1989}.  The features can enter through an additive structure with coefficients $\beta_j \in \mathbb{R}, j = 1,...,q$.  Equation (\ref{DeepModel}) includes all possible $q$ components (features), using binary variables $\gamma_j$ indicating whether the corresponding variables are to be actually included into the model or not.
Priors for the different parameters of the model are specified in Section \ref{sec:priors}.

\subsection{Topology of the feature space}\label{sec:feature}
The feature space is constructed through  a hierarchical  procedure similar to the deep learning approach~\citep{lecun2015deep, Goodfellow-et-al-2016}, but allowing for automatic construction of the architecture. This is in contrast to deep neural networks where the architecture has to be set in advance~\citep{lecun2015deep}. We can also obtain posterior distributions of different architectures within the suggested approach.

Deep neural networks typically use one or several pre-specified activation functions to compute hidden layer values, where currently the rectified linear unit $\text{ReLU}(x) = \max\{0, x\}$ is the most popular. The configuration of where (at which layers and for which subsets of neurons) these functions are applied is fixed. In contrast, in our approach the activation function $g$ can be dynamically selected from a pre-specified \emph{set} $\cal{G}$ of non-linear functions, which can include, apart from the rectified linear unit, several other transformations, possibly adjustable to the particular application. Examples include $\text{exp}(x)$, $\text{log}(x)$, $\text{tanh}(x)$, $\text{atan}(x)$ and $\text{sigmoid}(x)$.

The   construction of possible features will, similarly to deep neural networks, be performed recursively through non-linear combinations of previously defined features. Let $\mathcal{G}$ denote a set of $l$ non-linear functions,  $\mathcal{G} = \{g_1(x),...,g_l(x)\}$.
Define the \emph{depth} of a feature as the maximal number of nonlinear functions from $\mathcal{G}$ applied recursively when generating that feature. For example, a feature $F(x) = \sin\left(\cos\left(\log(x)\right)+\exp(x)\right)$ has depth equal to 3. Denote the set of features of depth $d$ by ${\cal F}_d$, which will be of size $q_d$. Furthermore denote the vector of all features of depth less or equal to $d$ by $\boldsymbol{F}^d(\boldsymbol{x})$.  The inner layer of features of depth zero consists of the original covariates themselves, ${\cal F}_0 := \{x_1, \dots, x_m\}$, where we drop the index $i$ for notational convenience.  Then $q_0 = m$ and $\boldsymbol{F}^0(\boldsymbol{x}) = \boldsymbol{x} = (x_1, \dots, x_m)$.  A new feature $F \in {\cal F}_{d+1}$ is obtained by applying a non-linear transformation $g$ on an affine transformation of  $\boldsymbol{F}^d(\boldsymbol{x})$:
\begin{equation} \label{eq:new_features}
F(\boldsymbol{x}) = g(\alpha_0+\boldsymbol{\alpha}^T \boldsymbol{F}^d(\boldsymbol{x})) \;.
\end{equation}
Equation~\eqref{eq:new_features} has the functional form most commonly used in deep neural networks models \citep{Goodfellow-et-al-2016}, though we allow for linear combinations of features from layers of \emph{different} depths as well as combinations of \emph{different} non-linear transformations.
The affine transformation in~\eqref{eq:new_features} is parameterized by the intercept $\alpha_0 \in \mathbb{R}$ and the coefficient vector of the linear combination $\boldsymbol{\alpha}$ which is typically very sparse but must include at least one non-zero coefficient  corresponding to a feature from ${\cal F}_{d}$. 
Different features of  ${\cal F}_{d+1}$  are distinguished only according to the \emph{hierarchical pattern} of non-zero entries of $\boldsymbol{\alpha}$.
This is similar to the common notion in variable selection problems where models including the same variables but different non-zero coefficients are still considered to represent the same model. In that sense our features are characterized by the model topology and not by the exact values of the coefficients $(\alpha_0,\boldsymbol{\alpha})$.

The number of features $q_d$ with depth of size $d$ can be calculated recursively with respect to the number of features from the previous layer, namely 
\begin{equation}\label{eq:num_features}
q_d=|\mathcal{G}|\left(2^{{\sum_{t=0}^{d-1}q_t}}-1\right)-\sum_{t=1}^{d-1}q_t,
\end{equation}
where as discussed above $q_0=m$ and $|\mathcal{G}|$ denotes the number of different functions included in $\mathcal{G}$. One can clearly see that
the number of features grows exponentially with depth.
In order to avoid potential overfitting through too complex models the two constraints are defined.
\begin{constraint}
 The depth of any feature involved is less or equal to $D_{max}$.
\end{constraint}
\begin{constraint}
The total number of features in a model is less or equal to $Q$.
\end{constraint}
\noindent
The first constraint ensures that the feature space is finite, with total size 
 $q = \sum_{d = 0}^{D_{max}} q_d$, while the second constraint limits the number of possible models by  $\sum_{k=1}^{Q}{q \choose k}$.
 The (generalized) universal approximation theorem~\citep{hurnik1991} is applicable to the defined class of models provided that $\cal{G}$ contains at least one bounded monotonously increasing function. Hence the defined class of models is dense in terms of approximating any function of interest in the closed domain of the Euclidean  space.

The feature space we have iteratively constructed through equation \eqref{eq:new_features} 
is extremely rich and encompasses as particular cases features from numerous other popular statistical and machine learning models. If the set of non-linear functions only consists of one specific transformation, for example $\mathcal{G} = \{\sigma(x)\}$
where $\sigma(\cdot)$ is the sigmoid function, then the corresponding feature space includes all possible neural networks with the sigmoid activation function. %It is still slightly larger because linear combinations of features of different depths are additionally possible. 
Another important class of models included in the DBRM framework are decision trees \citep{breiman1984classification}. Simple decision rules correspond to the non-linear function $g(x)=\text{I}(x\ge 1)$. 
Intervals and higher dimensional regions can be defined through multiplications of such terms. Multivariate adaptive regression splines \citep{friedman1991multivariate} are included by allowing a pair of piecewise linear functions $g(x) = \max\{0,x-t\}$ and $g(x) = \max\{0,t-x\}$. Fractional polynomials~\citep{royston1997approximating} can also be easily included through $g(x)=x^{r/s}$.
Logic regression, characterized by features being logic combinations of binary covariates~\citep{ruczinski2003logic,Hubin2017} 
is also fully covered by DBRM models. 
Combining more than one function in $\mathcal{G}$ provides quite interesting additional flexibilities in construction of features, e.g. $\left(0.5x_1+x_2^{0.5}+\text{I}(x_2>1)+\sigma(x_3)\right)^2$.

Interactions between (or multiplication of) variables are important features to consider. 
Assuming  that both $\log(x)$ and $\exp(x)$ are members of  $\mathcal{G}$,  multiplication of two (positive) features becomes a new feature with depth $d = 2$ via 
\begin{equation}
F_k(\boldsymbol{x}) * F_l(\boldsymbol{x}) = \exp\left(\log(F_k(\boldsymbol{x})) +\log(F_l(\boldsymbol{x}))\right)\;.\label{eq:mult}
\end{equation}
However, due to its importance  we will include the multiplication operator between two features directly in our algorithmic approach (see Section~\ref{section3}) and treat it simply as an additional transformation with depth 1. 

\subsection{Feature engineering} \label{Subsec:FeatureEngineering}

In principle one could generate any feature defined sequentially through~\eqref{eq:new_features} but in order to construct a computationally feasible algorithm for inference, restrictions on the choices of  $\boldsymbol\alpha$'s  are made. 
Our \textit{feature engineering} procedure computes specific values of $\boldsymbol\alpha$ for any engineered feature in the topology using the following three operators which take features of depth $d$ as input and create new features of depth $d+1$.
\[
F_j(\boldsymbol{x})=
\begin{cases}
g(F_k(\boldsymbol{x})) &\text{for a \textit{modification} of}\ F_k \in {\cal F}_d;\\
F_k(\boldsymbol{x})*F_l(\boldsymbol{x})  &\text{for a \textit{crossover} between}\ F_k \in {\cal F}_d, F_l \in {\cal F}_s (s \leq d); \\
g(\boldsymbol{\alpha}_j^T\B F^d(\boldsymbol{x}) + \alpha_{j,0}), &\text{for a nonlinear \textit{projection}}.
\end{cases}
\]

The \textit{modification} operator is the special case of~\eqref{eq:new_features} where $\boldsymbol \alpha_j$ has only one nonzero element $\alpha_{j,k}=1$. The \textit{crossover} operator can also be seen as a special case in the sense of~\eqref{eq:mult}. Only for the general \textit{projection} operator one has to estimate $\boldsymbol{\alpha}_j$, which is usually assumed to be very sparse. We have currently implemented four different strategies to compute $\boldsymbol{\alpha}$ parameters. Here and in Section \ref{section4} we focus on the simplest and computationally most efficient version. 
The other three strategies as well as further ideas to potentially improve the feature engineering step are discussed in Section~\ref{section5} and in Appendix~\ref{ap:add.ex1_3} available in the web supplement. 

Our default procedure to obtain $\boldsymbol{\alpha}_j$ is to compute maximum likelihood estimates for model~\eqref{themodeleq} including only  $F_{r_l}, r_l  = 1,...,w_j$ as covariates, that is for
\begin{equation} \label{Strategy_1}
\mathsf{h}(\mu) = \boldsymbol{\alpha}_j^T\B F^d(\boldsymbol{x})  + \alpha_{j,0} \; .
\end{equation}
 This choice is made not only for computational convenience, but also has some important advantages. 
The non-linear transformation $g$ is not involved when computing   $\boldsymbol{\alpha}_j$. Therefore the procedure can easily be applied for non-linear transformations $g$ which are not differentiable, like for example the extremely popular rectified linear unit of neural networks or the characteristic functions for decision trees. Furthermore ML estimation for generalized linear models usually involves convex optimization problems with unique solutions.  On the other hand this simple approach means that the parameters $\boldsymbol{\alpha}_{r_l}$
from  $F_{r_l}(\boldsymbol{x})$ are not re-estimated but kept fixed, a restriction which will be overcome by some of the alternative strategies (including a fully Bayesian one) introduced in Appendix~\ref{ap:add.ex1_3}.

\subsection{Bayesian model specifications} \label{sec:priors}

In order to put model \eqref{themodeleq} into a Bayesian framework one has to specify priors for all parameters involved. 
The structure of a specific model is uniquely defined by the vector $\M = (\gamma_1,\dots,\gamma_q)$. 
We introduce model priors which penalize for number and the \emph{complexity} of included features in the following way: 
\begin{align}
p(\M)\propto\ &
\prod_{j=1}^q a^{\gamma_jc(F_j(\B x))} \; .
\label{eq:modelprior}
\end{align}
The measure $c(F_j(\B x))\ge 0$ is a non-decreasing function of the complexity of feature $F_j(\B x)$. 
With $0<a<1$, the prior prefers both fewer terms and simpler features over more complex ones.

There are many different ways of defining feature complexity. 
We will consider a measure taking into account both the number of non-linear transformations and the number of features used at each transformation step. 
Define the \emph{local width} of a feature as the number of non-zero coefficients of $\boldsymbol{\alpha}$ (including $\alpha_0$) in equation~\eqref{eq:new_features}. Features obtained with a modification operator or a crossover operator both have local width 1.
However, features may inherit different widths from parental layers. Define accordingly the \emph{total width} of a feature recursively as the sum of all local widths of features contributing  in equation~\eqref{eq:new_features}.
In the current implementation of DBRM this total width serves as  complexity measure as illustrated in the following example. Consider a feature of the form
\[
F_3(\boldsymbol{x}) = g(\alpha_{3,0}+\boldsymbol{\alpha}^T_3(g(\alpha_{1,0}+\boldsymbol{\alpha}^T_{1}\boldsymbol{x}),g(\alpha_{2,0}+\boldsymbol{\alpha}^T_{2}\boldsymbol{x}))
\]
which has a depth of $d = 2$ and a total width of $w = ||\boldsymbol{\alpha}_{1}||_0+||\boldsymbol{\alpha}_{2}||_0+||\boldsymbol{\alpha}_3||_0+3$. Here $||\cdot||_0$ refers to the $l_0$-"norm" 
(the number of non-zero elements in the corresponding vector) and the additional 3 corresponds to the intercept terms.

To complete the Bayesian model one needs to specify priors for $\boldsymbol{\beta}^{\M}$, the vector of regression parameters   for which $\gamma_j = 1$ and, if necessary, for the dispersion parameter $\phi$. 
\begin{align}
\BB\beta^{\M}|\phi\sim&p(\boldsymbol{\beta}^{\M}|\phi) \; ,
\label{eq:prior.par}\\
\phi\sim& p_{\phi}(\phi).
\end{align}
Prior distributions on $\BB\beta$ and $\phi$
are usually defined in a way to facilitate efficient computation of marginal likelihoods  (for example by specifying conjugate priors) and should be carefully chosen for the applications of interest. Specific choices are described in Section \ref{section4} when considering different real data sets.

\subsection{Extensions of DBRM}

Due to the model-based approach of DBRM different kinds of extensions can be considered. One important 
extension is to include latent variables, both to take into account correlation structures and over-dispersion. Simply replace~\eqref{DeepModel}
by
\begin{align}
   \mathsf{h}(\mu_i) = & \beta_0  + \sum_{j=1}^{q} \gamma_{j}\beta_{j}F_{j}(\boldsymbol{x}_{i})+\sum_{k=1}^{r} \lambda_{k}\delta_{ik}\label{DeepModel2}
\text{ where }
\boldsymbol{\delta_k} = (\delta_{1k},...,\delta_{nk}) \sim N_{n}\left(\boldsymbol{0}, \boldsymbol\Sigma_k\right).
\end{align}
In this formulation of the DBRM model, equation~\eqref{DeepModel2} now includes $q + r$ possible components where $\lambda_k$ indicates whether the corresponding latent variable is to be included into the model or not.
The latent Gaussian variables with covariance matrices $\boldsymbol{\Sigma}_k$ allow to describe different correlation structures between individual observations (e.g. autoregressive models). The matrices typically depend only on a few parameters, so that in practice one has $\boldsymbol{\Sigma}_k=\boldsymbol{\Sigma}_k(\boldsymbol\psi_k)$.

The model vector now becomes $\M= (\BB\gamma,\BB\lambda)$ where $\BB\lambda=(\lambda_1,...,\lambda_r)$.
Similar to the restriction on the number of features that can be included in a model, we introduce an upper limit $R$ on the number of latent variables.
The total number of models with non-zero prior probability will then be $\sum_{k=1}^{Q}{q \choose k} \times \sum_{l=1}^{R}{r \choose l}$.
The corresponding prior for the model structures is defined by
\begin{align}
p(\M)\propto\ &
\prod_{j=1}^q a^{\gamma_jc(F_j(\B x))}
\prod_{k=1}^r  b^{\lambda_kv(\delta_k)}.\label{eq:modelprior2}
\end{align}
Here the function  $v(\delta_k)\ge 0$ is a measure for the complexity of the latent variable $\delta_k$,   which is assumed to be a non-decreasing function of the number of hyperparameters defining the distribution of the latent variable. In the current implementation we simply count the number of hyperparameters. 
The prior is further extended to include
\begin{align}
\boldsymbol{\psi}_k\sim&\pi_k(\boldsymbol{\psi}_k), &&\text{for each $k$ with $\lambda_k=1$}.\label{latentprior}
\end{align}

\subsection{Bayesian inference}\label{sec:Bayes.inference}

Posterior marginal probabilities for the model structures are, through 
Bayes formula, given by
\begin{equation}\label{PMP}
p(\M|\B y)
=\frac{p(\M)p(\B y|\M)}
      {\sum_{\M'\in\MS}
      p(\M')p(\B y|\M')}\; ,
\end{equation}
where $p(\By|\M)$ denotes the marginal likelihood of $\By$ given a specific model $\M$. Due to the huge size of $\MS$ it is not possible to calculate the sum in the denominator of~\eqref{PMP} exactly.
In Section~\ref{section3} we will discuss how to obtain estimates of $\hat p(\M|\B y)$ using MCMC algorithms.

The (estimated) posterior distribution of any statistic $\Delta$ of interest (like for example in predictions) becomes 
\begin{equation}\label{posterior_quantile}
\hat p(\Delta|\By) =  \sum_{\M \in\MS}{p(\Delta|\M,\By)\hat p(\M|\By)} \; .
\end{equation}
The corresponding expectation is obtained via model averaging:
\begin{equation}\label{expectation_quantile}
\hat{\text{E}}[\Delta|\By] =  \sum_{\M \in\MS}{\text{E}[\Delta|\M,\By]\hat p(\M|\By)} \;.
\end{equation}
An important example is 
the posterior marginal inclusion probability of a specific feature $F_j(\boldsymbol{x})$, which can be estimated by
\begin{equation}\label{marginal_inclusion}
\hat p({\gamma}_{j}=1|\By) =  \sum_{\M:\gamma_j=1}\hat p(\M|\By).
\end{equation}
This provides a well defined measure of the importance of an individual component.

\section{Fitting of DBRM}\label{section3}
In this section we will develop algorithmic approaches for fitting the DBRM model. The main tasks are to (i) calculate the marginal likelihoods
$p(\B y|\M)$ for a given model and (ii) to search through the model space $\MS$. Concerning the first issue one has to solve the integral 
\begin{align} \label{MarginalPosterior}
p(\B y|\M)
= & \int_{\BB\theta}p(\B y|,\BB\theta_{\M},\M)p(\BB\theta_{\M}|\M)d\BB\theta_{\M}
\end{align}
where $\BB\theta_{\M}$ are all the parameters involved in model $\M$.
In general these marginal likelihoods are difficult to calculate.
%~\citep{Friel2012,HubinStorvikINLA}. 
Depending on the model specification we can either use exact calculations when these are available~\citep{Clyde:Ghosh:Littman:2010} or numerical approximations based on simple Laplace approximations~\citep{tierney1986accurate}, the popular integrated nested Laplace approximation (INLA) \citep{rue2009eINLA} or MCMC based
methods like Chib's or Chib and Jeliazkov's method \citep{chib1995marginal,chib2001marginal}. Some comparison of these methods are presented in~\citet{Friel2012} and~\citet{HubinStorvikINLA}. 

The parameters $\boldsymbol\theta_{\M}$ of DBRM for a specified model topology  consist of $\boldsymbol\beta_{\M}$, the regression coefficients for the features, and $\phi$, the dispersion parameter. If latent Gaussian variables are included into the models,  parameters $\boldsymbol\psi_{\M}$ will also be part of $\boldsymbol\theta_{\M}$.
We are not including here the set of coefficients $\boldsymbol\alpha_{\M}$ which encompasses all the parameters inside the features $F_j(\boldsymbol{x})$ of model $\M$. These are considered simply as constants, used to iteratively generate features of depth $d$ as described in Section \ref{Subsec:FeatureEngineering}.  One \emph{can} make the model more general and consider  $\boldsymbol\alpha_{\M}$ as part of the parameter vector $\boldsymbol\theta_{\M}$. However, solving the integral \eqref{MarginalPosterior} over the full set of parameters $\boldsymbol\theta_{\M}$ including $\boldsymbol\alpha_{\M}$ will become computationally extremely demanding due to the complex non-linearity and the  high-dimensional integrals involved. Some possibilities how to tackle this problem in the future are portended in  Section \ref{section5}.
%\todo{FF: I have moved the relevant paragraph to the Discussion. It still needs some rewriting.}

%\subsection{Search algorithms} \label{Sec:GMJMCMC}

Condsider now task (ii), namely the development of algorithms for searching through the model space. 
Calculation of  $p(\M|\By)$ requires to iterate through the space $\MS$ including all potential models, which due to the combinatorial explosion \eqref{eq:num_features} becomes computationally infeasible for  even a moderate set of input variables and latent Gaussian variables. We therefore aim at approximating $p(\M|\By)$ by means of finding a subspace $\MS^* \subset \MS$ which can be used to approximate~\eqref{PMP} by 
\begin{equation}
\widehat{p}(\M|\By) =  
\frac{p(\M)p(\By|\M)}
       {\sum_{\M' \in \MS^*}{p(\M')p(\By|\M')}}\,\text{I}(\M\in\MS^*) \; .
 \label{approxpost}
\end{equation}
Low values of $p(\M)p(\By|\M)$ induce both  low values of the numerator and small contributions to the denominator in~\eqref{PMP}, hence models with low mass $p(\M)p(\By|\M)$ will have no significant influence on posterior marginal probabilities for other models. On the other hand, models with high values of $p(\M)p(\By|\M)$ are important to be addressed. It might be equally important to include \emph{regions} of the model space where no single model has particularly large mass but there are many  models giving a moderate contribution.  Such regions of high posterior mass are particularly important for constructing a reasonable subspace $\MS^* \subset  \MS$ and missing them can dramatically influence our posterior estimates. 

The mode jumping MCMC algorithm  (MJMCMC) was introduced in~\citet{Hubin2016}
for variable selection within standard GLMM models, that is models where all possible features are pre-specified.
The main ingredient in MJMCMC is the specification of (possibly large) moves in the model space. This algorithm was generalized to the genetically modified MJMCMC algorithm (GMJMCMC) in the context of logic regression by~\citet{Hubin2017}. The GMJMCMC is not a proper MCMC algorithm in the sense of converging to the posterior distribution $p(\M|\BB y)$ although it does provide consistent model estimates by means of the approximation~\eqref{approxpost}. 
In the following two subsections we are suggesting an adaptation of the GMJMCMC algorithm to DBRM models. Additionally, we derive a fully reversible GMJMCMC algorithm (RGMJMCMC). Since both algorithms rely on the MJMCMC algorithm we start with a short review of this algorithm. Throughout this section without loss of generality we will only consider features and not latent variables. Selection of latent variables is part of the implemented algorithms but only complicates the presentation.

\subsection{The mode jumping MCMC algorithm}

Consider the case of a fixed predefined set of
$q$ potential features with no latent variables. 
Then the general model space $\MS$ is of size $2^q$ and standard MCMC algorithms
tend to get stuck in local maxima. The mode jumping MCMC procedure (MJMCMC) was originally proposed by~\citet{Tjelmeland99modejumping} for continuous space problems and recently extended to model selection settings by~\citet{Hubin2016}. MJMCMC is a proper MCMC algorithm equipped with the possibility to jump between different modes within the discrete model space. 
The algorithm is described in Algorithm~\ref{MJMCMCalg0}.
\begin{algorithm}[t]
\caption{\label{MJMCMCalg0}MJMCMC}
\begin{algorithmic}[1]
\State  Generate a large jump  $\M_0^*$ according to a  proposal distribution  $q_l(\M_{(0)}^*|\M)$.
\item Perform a local  optimization, defined through $\M_{(k)}^*\sim q_o(\M_{(k)}^*|\M_{(0)}^*)$.
\State Perform a small randomization  to generate the proposal $\M^*\sim q_r(\M^*|\M_{(k)}^*)$.
\item  Generate backwards auxiliary variables $\M_{(0)}\sim q_l(\M_{(0)}|\M^*)$,  $\M_{(k)}\sim q_o(\M_{(k)}|\M_{(0)})$.
\State Put
\[\M'=
\begin{cases}\M^*&\text{with probability $r_{mh}(\M,\M^*;\M_{(k)},\M_{(k)}^*)$;}\\
\M&\text{otherwise,}
\end{cases}
\]
where
\begin{equation}
r_{mh}^*(\M,\M^*;\M_{(k)},\M_{(k)}^*) = \min\left\{1,\frac{\pi(\M^*)q_r(\M|\M_{(k)})}{\pi(\M)q_r(\M^*|\M_{(k)}^*)}\right\}\label{locmcmcgen00}.
\end{equation}
\end{algorithmic}
\end{algorithm}
The basic idea is to make a large jump (changing many model components) combined with local optimization within the discrete model space to obtain a proposal model with high posterior probability. Within a Metropolis-Hastings setting a valid acceptance probability is constructed using symmetric backward kernels, which guarantees that the resulting Markov chain is ergodic and has the desired limiting distribution \citep[see][for details]{Hubin2016}.

\subsection{Genetically Modified MJMCMC}

The MJMCMC algorithm  
requires that all the covariates (features) defining the model space are known in advance and are all considered at each iteration of the algorithm. For the DBRM models, the features are of a complex structure and a major problem in this setting is that it is quite difficult to fully specify the space $\MS$ in advance (let alone storing all potential features in some data structure). 
The idea behind GMJMCMC is to apply the MJMCMC algorithm within a smaller set of model components in an iterative setting. 

\subsubsection{Main algorithm}
Throughout our search we generate a sequence of so called \emph{populations} $\mathcal{S}_1,\mathcal{S}_2,...,\mathcal{S}_{T_{max}}$. Each  $\mathcal{S}_t$ is a set of $s$ features and forms a separate \textit{search space} for exploration through MJMCMC iterations.
Populations dynamically evolve allowing GMJMCMC to explore different parts of the total model space.
Algorithm~\ref{gMJMCMCalg} summarizes the procedure, the exact generation of $\mathcal{S}_{t+1}$ given $\mathcal{S}_t$ is described below. 
\begin{algorithm}[t]
\caption{GMJMCMC}\label{gMJMCMCalg}
\begin{algorithmic}[1]
\State Initialize $\mathcal{S}_{0}$
\State Run the MJMCMC algorithm within search space $\mathcal{S}_0$ for $N_{init}$ iterations and use results to initialize $\mathcal{S}_1$.
\For{$t=1,...,T_{max - 1}$} 
\State Run the MJMCMC algorithm within search space $\mathcal{S}_t$ for $N_{expl}$ iterations.\State Generate a new population $\mathcal{S}_{t+1}$
\EndFor
\State Run the MJMCMC algorithm within search space $\mathcal{S}_{T_{max}}$ for $N_{final}$ iterations.
\end{algorithmic}
\end{algorithm}
   
The following result is concerned with consistency of probability estimates of GMJMCMC when the number of iterations increases. The theorem is an adaption of Theorem~1 in~\citet{Hubin2017}:
\begin{theorem}\label{th:GMJMCMC}
Let $\MS^*$ be the set of models visited through the GMJMCMC algorithm.   Define $M_{S_t}$ to be the set of models visited at iteration $t$ within search space $\mathcal{S}_t$. Assume $s\ge Q$ and  
$\{(\mathcal{S}_t,M_{\mathcal{S}_t})\}$ forms an
irreducible Markov chain over the possible states.
Then the model estimates based on~\eqref{approxpost} will converge to the true model  probabilities as the number of iterations $T_{\max}$ converges to $\infty$.
\end{theorem}
\begin{proof}
Note that the approximation~\eqref{approxpost} will provide the exact answer if $\MS^*=\MS$. It is therefore enough to show that the algorithm in the limit will have visited all possible models. 
Since the state space of the irreducible Markov chain $\{(\mathcal{S}_t,M_{\mathcal{S}_t})\}$  is finite, it is also recurrent, and there exists a stationary distribution with positive probabilities on every model. Thereby, all states, including all possible models of maximum size $s$, will eventually be visited.
\end{proof}

\paragraph{Remark} All models visited, also those auxiliary  ones which are used by MJMCMC to generate proposals, will be included into $\MS^*$. For these models, also computed marginal likelihoods will be stored,
making the costly likelihood calculations only necessary for models that have not been visited before.

\subsubsection{Initialization}
The algorithm is initialized by first applying some  procedure (for example based on marginal testing) which selects a subset of $q_0 \leq q$ input covariates.
We denote these preselected components 
by $\mathcal{S}_0 = \{\boldsymbol{x}_1,...,\boldsymbol{x}_{q_0}\}$ 
where for notational convenience we have ordered indices according to preselection which does not impose any loss of generality. Note that depending on the initial preselection procedure, $\mathcal{S}_0$ might include a different number of components than all further populations $\mathcal{S}_t$. MJMCMC is then run for a given number of iterations $N_{init}$  on $\mathcal{S}_0$ and the resulting $s_1<s$ input components 
with highest frequency (ratio of models including the component) 
will become the first $s_1$ members 
of population $\mathcal{S}_1$.
The remaining $s-s_1$ members of  $\mathcal{S}_1$ will be newly created features generated by applying the transformations described in Section \ref{populch} on members of $\mathcal{S}_0$.

\subsubsection{Transition between populations}\label{populch}%

Members of the new population $\mathcal{S}_{t+1}$ are generated by applying certain transformations to components of $\mathcal{S}_t$. 
First some components with low frequency from search space $\mathcal{S}_t$  are removed using a $filtration$ operator. The removed components are then replaced, where each replacement  is generated randomly by a \textit{mutation} operator with probability $P_m$, by a \textit{crossover}  operator with probability $P_c$, by a \textit{modification} operator with probability $P_t$ or by a \textit{projection} operator with probability $P_p$, where $P_c+P_m+P_t+P_p = 1$ (adaptive versions of these probabilities are considered in section~\ref{tricks}). The operators to generate potential features of $\mathcal{S}_{t+1}$ are formally defined below,
where the modification, crossover and projection operators have been introduced already in Section \ref{sec:feature}. Combined, these possible transformations fulfill the requirement about irreducibility in Theorem~\ref{th:GMJMCMC}.

\begin{description}
\item[Filtration:] Features  within $\mathcal{S}_t$ with estimated posterior probabilities below a given threshold are deleted  with probability $P_{del}$. The algorithm offers the option that a subset of $\mathcal{S}_0$  is always kept in the population throughout the search. 

\item[Mutation:] A new  feature  $F_{k}$ is randomly selected from $\mathcal{F}_0$. 
\item[Modification:] A new feature $F_k=g(F_j)$ is created where $F_j$ is randomly selected from $\mathcal{S}_t\cup\mathcal{F}_0$.
%according to the probability distribution $P_{p,t}$ 
and $g(\cdot)$ is randomly selected from $\mathcal{G}$.
%is selected according to the probability distribution $P_{g,t}$. 

\item[Crossover:] A new feature $F_{j_1}* F_{j_2}$ is created by randomly
selecting  $F_{j_1}$ and $F_{j_2}$ from $\mathcal{S}_t\cup\mathcal{F}_0$. %according to some probability distributions $P_{p,t}$.

\item[Projection:] A new feature $F_k=g\left(\alpha_0+\boldsymbol{\alpha}^T\B F^*\right)$ is created in three steps. First a (small) subset of
$\mathcal{S}_t$ is selected by sampling without replacement.
Then $(\alpha_0,\boldsymbol{\alpha})$ is specified according to the rules described in Section~\ref{Subsec:FeatureEngineering} and finally  $g$ is randomly selected   from~$\mathcal{G}$.
\end{description}
\paragraph{} 
 For all features generated with any of these operators it holds that if either the newly generated feature is already present in $\mathcal{S}_t$ or it has linear dependence with the currently present features then it is not considered for  $\mathcal{S}_{t+1}$. In that case a different feature is generated as just described.

\subsubsection{Reversible Genetically Modified MJMCMC}

The GMJMCMC algorithm described above is not reversible and hence cannot guarantee that the ergodic distribution of its Markov chain corresponds to the target distribution of interest \citep[see][for more details]{Hubin2017}. An easy modification based on performing both forward and backward swaps between populations can provide a proper MCMC algorithm in the model space of DBRM models.
Consider a  transition $\M \rightarrow \mathcal{S}' \rightarrow \M'_0 \rightarrow ... \rightarrow \M'_k \rightarrow \M'$ with a given probability kernel. Here $q(\mathcal{S}'|\M)$  is the proposal for a new population, 
transitions  $\M'_0 \rightarrow ... \rightarrow \M'_k$ are generated by local MJMCMC within the model space induced by $\mathcal{S}'$, and the transition $\M'_k \rightarrow \M'$ is some randomization at the end of the procedure as described in the next paragraph. 
The following theorem shows the detailed balance equation for the suggested swaps between models.
\begin{theorem}
Assume  $\M\sim p(\cdot|\B y)$ and $(\mathcal{S'},\M_k',\M')$ are generated according to the large jump proposal distribution $q_S(\mathcal{S'}|\M)q_o(\M_k'|\mathcal{S'},\M)q_r(\M'|\mathcal{S},\M_k')$. Assume
further $(\mathcal{S},\M_k)$ are generated according to $\tilde{q}_S(\mathcal{S}|\M',\mathcal{S},\M)q_o(\M_k|\mathcal{S},\M')$. Put
\[
\M^*=\begin{cases}
\M'&\text{with probability $\min\{1,a_{mh}\}$;}\\
\M&\text{otherwise.}
\end{cases}
\]
where
\begin{equation}\label{eq:a.rjmcmc}
a_{mh} = \frac{p(\M'|\B y)q_r(\M|\mathcal{S},\M_k)}{p(\M|\B y)q_r(\M'|\mathcal{S'},\M'_k)}.
\end{equation}
Then  $\M^*\sim p(\cdot|\B y)$.
\end{theorem}

\begin{proof}
Define
\[
\bar{p}(\M,\mathcal{S'},\M_k')\equiv p(\M|\B y)q_\mathcal{S}(\mathcal{S'}|\M)q_o(\M_k'|\mathcal{S'},\M).
\]
Then by construction $(\M,\mathcal{S'},\M_k')\sim \bar{p}(\M,\mathcal{S},\M_k')$.
Define $(\M',\mathcal{S},\M_k)$ to be a proposal from the distribution  $q_r(\M'|\mathcal{S},\M_k')q_S(\mathcal{S}|\M')q_o(\M_k|\mathcal{S},\M')$.
Then the Metropolis Hastings acceptance ratio becomes
\[
\frac{\bar{p}(\M',\mathcal{S},\M_k)q_r(\M|\mathcal{S},\M_k)q_S(\mathcal{S'}|\M)q_o(\M_k'|\mathcal{S'},\M)}{\bar{p}(\M,\mathcal{S'},\M_k')q_r(\M'|\mathcal{S}',\M_k')q_S(\mathcal{S}|\M')q_o(\M_k|\mathcal{S},\M')}
\]
which reduces to $a_{mh}$. 
\end{proof}

From this theorem it follows that if the Markov chain is irreducible in the model space then it is ergodic and converges to the right posterior distribution.
The described procedure marginally generates samples from the target distribution, i.e. according to model posterior probabilities $p(\M|\By)$. Note that the populations themselves do not have to be stored, they are only needed for the generation of new models. Instead of using the approximation~\eqref{approxpost}  one can get frequency based estimates of the model posteriors $p(\M|\By)$. For a sequence of simulated models $\M^1,...,\M^W$ from an ergodic MCMC algorithm with $p(\M|\By)$ as a stationary distribution  it holds that
\begin{equation}\label{map2}
\widetilde{p}(\M|\By)=W^{-1}\sum_{i=1}^W \text{I}(\M^{(i)} = \M) \xrightarrow[W\rightarrow\infty]{d} p(\M|\By) \; 
\end{equation}
and similar results are valid for estimates of the posterior marginal inclusion probabilities~\eqref{marginal_inclusion}.
%In  general the  estimates~\eqref{approxpost} tend to be  more accurate, as demonstrated by \cite{Hubin2016}.

Proposals $q_\mathcal{S}(\mathcal{S}'|\M)$  are obtained as follows.
First all members of $\M$ are included. Then additional
features are added  similarly as described in Section~\ref{populch} but with $\mathcal{S}_t$ replaced by the population including all components in $\M$. An adaptive version of this can be achieved by
dynamically changing $\mathcal{F}_0$ to include all features that previously have been considered,  the validity of which is explained in Section~\ref{tricks}.

The randomization $\M'\sim q_r(\M|\mathcal{S}',\M_k')$ is performed by possible swapping of the features within $\mathcal{S}'$, each with a small probability $\rho_r$.
Note that this might give a reverse probability
$q_r(\M|\mathcal{S},\M_k)$ being zero if $\mathcal{S}$ does not include all components in $\M$. In that case the proposed model is not accepted. Otherwise 
the ratio of the proposal probabilities can be written as $\frac{q_r(\M|\mathcal{S},\M_k)}{q_r(\M'|\mathcal{S}',\M'_k)} = \rho_r^{d(\M,\M_k)-d(\M',\M'_k)}$, where $d(\cdot,\cdot)$ is the Hamming distance (the number of components differing).

\subsection{Important computational tricks}\label{tricks}
To make the algorithms work sufficiently fast our implementation includes several tricks to be described below. 
\subsubsection*{Delayed rejection}
In order to make the computations more efficient and avoid  unnecessary backward searches we make use of the so called delayed acceptance approach. 
The most computationally demanding parts of the RGMJMCMC algorithms are the forward and backward MCMC searches (or optimizations). Often the proposals generated by forward search have a very small probability $\pi(\M')$ resulting in a low acceptance probability regardless of the way the backwards auxiliary variables are generated. In such cases, one would like to reject directly without performing the backward search. This can be achieved by the delayed acceptance procedure~\citep{christen2005markov,banterle2015accelerating} which can be applied in our case due to the following result:
\begin{theorem}
Assume $\M\sim p(\cdot|\B y)$ and $\M'$ is generated according to the RGMJMCMC algorithm. Accept $\M'$ if
both
\begin{enumerate}
\item $\M'$ is preliminarily accepted with a probability $\min\{1,\tfrac{p(\M'|\B y)}{p(\M|\B y)}\}$
\item and is finally accepted with a probability 
$\min\{1,\tfrac{q_r(\M|\mathcal{S},\M'_k)}{q_r(\M'|\mathcal{S}',\M_k)}\}$.
\end{enumerate}
Then also $\M \sim p(\cdot|\B y)$.
\end{theorem}

\begin{proof}
It holds for $a_{mh}$ given by~\eqref{eq:a.rjmcmc} that
\begin{align*}
a_{mh}(\M,\mathcal{S}',\M_k';\M',\mathcal{S},\M_k)
=&  a_{mh}^1(\M,\mathcal{S}',\M_k';\M',\mathcal{S},\M_k)\times a_{mh}^2(\M,\mathcal{S}',\M_k';\M',\mathcal{S},\M_k)
\intertext{where}
a_{mh}^1(\M,\mathcal{S}',\M_k';\M',\mathcal{S},\M_k)=&\frac{p(\M'|\B y)}{p(\M|\B y)},\quad
a_{mh}^2(\M,\mathcal{S}',\M_k';\M',\mathcal{S},\M_k)=\frac{q_r(\M|\mathcal{S},\M'_{k})}{q_r(\M'|\mathcal{S}',\M_k)}
\end{align*}
Since $a_{mh}^j(\M,\mathcal{S}',\M_k';\M',\mathcal{S},\M_k)
=[a_{mh}^j(\M',\mathcal{S},\M_k;\M,\mathcal{S},\M'_k)]^{-1}$ for $j=1,2$, it follows by the general results in~\citet{banterle2015accelerating} that we obtain an invariant kernel with respect to the target distribution.
\end{proof}

In general delayed acceptance results in a decreased total acceptance rate ~\citep[][remark 1]{banterle2015accelerating}, but it is still worthwhile due to the computational gain by avoiding the backwards search step in case of preliminary rejection.  Delayed acceptance is implemented in the RGMJMCMC algorithm of our R-package and is used in the examples of Section \ref{section4}.

\subsubsection*{Adaptive proposals}

Another important trick consists of using the chain's history to approximate marginal inclusion probabilities and utilize the latter for the proposal of new populations. Using any measure based on the marginal inclusion probabilities is valid for the reversible algorithm by Theorem 1 from \citet{roberts2007} for which two conditions need to be satisfied:
\begin{description}
\item [{\rm \it Simultaneous uniform ergodicity:}] For any set of tuning parameters, the Markov chain should be ergodic. \\
In our case, we have a finite %and countable 
number of models in the model space and hence their enumeration can be performed theoretically within a limited time provided irreducibility of the algorithm. Hence the first condition of the theorem is satisfied.

\item [{\rm \it Diminishing adaptation:}] The difference between following transition probabilities converge to zero.\\
As long as the inclusion probabilities that are used are truncated away from 0 and 1 by a small value $\varepsilon$, the frequencies will converge to the true marginal inclusion probabilities and the diminishing adaptation condition is satisfied. This is again possible because of the irreducibility of the constructed Markov chains in the finite model space.

\end{description}
\subsubsection*{Parallelization strategy}
Due to our interest in exploring as many \emph{unique}  high quality models as possible and doing it as fast as possible, running multiple parallel chains is likely to be computationally beneficial compared to running one long chain. The process can be embarrassingly parallelized into $B$ chains. If one is mainly interested in model probabilities, then equation~\eqref{approxpost} can be directly applied with $\MS^*$ now being the set of unique models visited within all runs. A  more memory efficient alternative is to utilize the following posterior estimates based on weighted sums over individual runs:
\begin{equation}\label{weighted_sum}
\tilde{p}(\Delta | \By) = \sum_{b=1}^B u_b\tilde{p}_b(\Delta | \By)\;.
\end{equation}
Here $u_b$ is a set of arbitrary normalized weights and $\tilde{p}_b(\Delta |\By)$ are the posteriors obtained with either equation~\eqref{approxpost} or~\eqref{map2} from run $b$ of GMJMCMC or RGMJMCMC. Due to the irreducibility of the GMJMCMC procedure it holds that $\lim_{k\rightarrow \infty}\tilde{p}(\Delta| \By) = p(\Delta| \By)$ where $k$ is the number of iterations. Thus for any set of normalized weights the approximation $\tilde{p}(\Delta| \By)$ converges to the true posterior probability ${p}(\Delta| \By)$ and one could use for example $u_b = \frac{1}{B}$. However, uniform weights have the disadvantage of potentially giving too much weight to posterior estimates from  chains that have not quite converged. In the following heuristic improvement $u_b$ is chosen to be proportional to the posterior mass detected by run $b$, 
\begin{align*}
u_b=&\frac{\sum_{\M \in {\MS^{*}}_b} p(\By| \M) p(\M)}{\sum_{b=1}^B\sum_{\M' \in {\MS^{*}}_b } p(\By| \M') p(\M')}\; .
\end{align*}
This choice  indirectly penalizes chains that cover smaller portions of the model space. When estimating posterior probabilities using these weights we only need, for each run,  to store the following quantities: $\tilde{p}_b(\Delta| \By)$ for all statistics $\Delta$ of interest and $s_b = \sum_{\M' \in {\MS^{*}}_b} p(\By|\M') p(\M')$ as a \textit{'sufficient'} statistic of the run. There is no further need of data transfer between processes. A proof that this choice of weights gives consistent estimates of posterior probabilities is given in \citet{Hubin2017}.

\section{{Applications}}\label{section4}
In this section we will first present three examples addressing prediction in the classification setting, where the  performance of  DBRM with GMJMCMC and RGMJMCMC 
is compared with nine competing algorithms. 
Then we present two examples 
of model inference after fitting  deep regression models with GMJMCMC and RGMJMCMC.  
Additionally two examples are presented in the Appendix, where the first one considers data simulated using a logic regression model and the second one illustrates the extended DBRM including latent Gaussian variables to analyze epigenetic data.

\subsection{Prediction}

\newcounter{Example} 

The first three examples of binary classification use the following  publicly available data sets: NEO objects data from NASA Space Challenge 2016 ~\citep{Neodata}, a breast cancer data set~\citep{breastdata} and some data concerned with spam emails \citep{cranor1998spam}. The performance of DBRM is compared with the following competitive algorithms:  tree based (TXGBOOST) and linear (LXGBOOST) gradient boosting machines, elastic networks (LASSO and RIDGE), deep dense neural networks with multiple hidden fully connected layers (DEEPNETS), random forest (RFOREST), naive Bayes (NBAYES), and simple \textit{frequentist} logistic regressions (LR). The corresponding R libraries, functions and their parameters settings are given in supplementary scripts.
%Appendix~\ref{appref} of the article. Parameter settings that were used for these algorithms are provided in Table~\ref{TunParam} in the Appendix~\ref{ap:tuning}.   

DBRM is fitted using either the GMJMCMC algorithm (DBRM{\_}G) or the reversible version (DBRM{\_}R), where, additionally to the standard algorithms parallel versions using $B=32$ threads were applied (DBRM{\_}G{\_}PAR and DBRM{\_}R{\_}PAR). 
%Here the index $x = 1,\dots,4$ refers to the 4 strategies of feature engineering described in  Section \ref{Subsec:FeatureEngineering}. 
For all classification examples the set of non-linear transformations is defined as  $\mathcal{G}=\{\text{gauss}(x),\text{tanh}(x),\text{atan}(x),\text{sin}(x)\}$, with $\text{gauss}(x) = e^{-x^2}$.
Additionally  a DBRM model with maximum depth $D_{max}=0$ (LBRM) is included, which corresponds to a linear Bayesian model  using only the original covariates.

Within DBRM, we apply logistic regression with independent observations, namely: 
\begin{subequations}\label{classification_modeleq}
\begin{align} 
  Y_i|\rho_i \sim& \ \text{Binom}(1,\rho_i) \\
   \text{logit}(\rho_i)  =& \ \beta_0  + \sum\limits_{j=1}^{q} \gamma_{j}\beta_{j}F_{j}(\boldsymbol{x}_{i}).
\end{align}
\end{subequations}
The Bayesian model
uses the model structure prior~\eqref{eq:modelprior}  with $a=e^{-2}$ and $Q = 20$. The resulting posterior corresponds to performing model selection with a criterion whose penalty on the complexity is similar to the AIC criterion, which is known (at least for the linear model) to be asymptotically optimal in terms of prediction~\citep{BA_2002}.
%and guarantees some optimality in terms of leave one out cross validation\todo{Geir: Include a reference? \\FF: I am not sure about the second half of the sentence. I think it is rather the case that AIC and LOOCV are asymptotically equivalent. Perhaps we might discard the second part of the sentence. We could simply cite the book by Burnham and Anderson (2002). If you want to be more specific cite the following two papers:\\Shibata,  R.  (1983)  Asymptotic  mean  efficiency  of  a  selection  of  regression  variables.
%Ann.  Inst. Statisti. Math. 35 , 415-423. \\
%Shao, J. (1997) An asymptotic theory for linear model selection (with discussion).
%Statistica Sinica,7, 221-242.}. 

The logistic regression model does not have a dispersion parameter and the Bayesian model is completed by using Jeffrey's prior for the regression parameters
\begin{align*}
%  p(\boldsymbol\gamma)\propto &  \prod_{j=1}^{q} \exp(-\gamma_j2c(F_j(\B x)))\label{gammodel2}\\
 p(\boldsymbol{\beta}^{\M})=&|J^{\boldsymbol{\gamma}}_n(\boldsymbol{\beta}^{\M})|^{\frac{1}{2}}\;.
\end{align*}
Here 
%$c(F_j(\B x))$ is the complexity of feature $j$. Finally, $\boldsymbol{\hat\beta}$ are the maximum likelihood estimates of the coefficients and 
$|J^{\boldsymbol{\gamma}}_n(\boldsymbol{\beta}^{\M})|$ is the determinant of the  Fisher information matrix. 

Predictions based on DBRM are made according to
$$
\hat{y}_i^* = \text{I}\left(\hat p(Y_i^*=1|\B y)\ge \eta \right),
$$
where we have used the notation $Y_i^*$ for a response variable in the test set. Furthermore
$$
\hat p(Y_i^*=1|\B y)=\sum_{\M \in \MS*} \hat p(Y_i^*=1|\M,\By) \hat p(\M|\By)
$$
with $\MS*$ denoting the set of all explored models and 
$$
\hat p(Y_i^*=1|\M,\By)= p(Y_i^*=1|\M,\widehat{\boldsymbol{\beta}}^{\M},\By)
$$
where $\widehat{\boldsymbol{\beta}}^{\M}$ is the
posterior mode in $p(\boldsymbol{\beta}^{\M}|\M,\B y)$. The most common threshold for  prediction is $\eta  = 0.5$. 
%\todo{Would it not better to say that choosing this prior is then equivalent to applying Laplace's approximation?} 
Calculation of marginal likelihoods
are performed through the Laplace approximation. 
%
%of the marginal likelihood \citep{Hubin2017} are considered.

To evaluate the predictive performance of algorithms we report the accuracy of predictions (ACC), false positive rate (FPR) and false negative rate (FNR), defined as follows:
\begin{align*}
\text{ACC} =& \frac{\sum_{i=1}^{n_p}\text{I}(\hat y_i^*=y_i^*) }{n_p};\\
\text{FPR} =& \frac{\sum_{i=1}^{n_p} \text{I}\left(y_i^*=0,\hat y_i^*=1\right)}{\sum_{i=1}^{n_p}  \text{I}\left(y_i^* = 0\right)};\\
\text{FNR} =& \frac{\sum_{i=1}^{n_p}  \text{I}\left(y_i^*=1,\hat y_i^*=0\right)}{\sum_{i=1}^{n_p} \text{I}\left(y_i^*=1\right)}.
\end{align*}
Here $n_p$ is the size of the test data sample.  For algorithms with a stochastic component, $N = 100$ runs were performed in the training data set and the test set was analysed with each of the obtained models, where we kept the split between training and test samples fixed. We then report the median as well as the minimum and maximum of the evaluation measures across those runs. For deterministic algorithms only one run was performed.

\refstepcounter{Example} \label{Ex:NeoAsteroids}
\subsubsection*{Example \arabic{Example}: Neo asteroids classification}
The dataset consists of characteristic measures of 20766 asteroids, some of which are class\-ified as potentially hazardous objects (PHO), whilst others are not. Measurements of the following nine explanatory variables are available: \textit{Mean anomaly, Inclination, Argument of perihelion, Longitude of the ascending node, Rms residual, Semi major axis, Eccentricity, Mean motion, Absolute magnitude}. 

\begin{table}[tbh]{
\caption{\label{t1}Comparison of performance (ACC, FPR, FNR) of different algorithms for NEO objects data. For methods with random outcome the median measures (with minimum and maximum in parentheses) are displayed.
The algorithms are sorted according to median accuracy.}
\resizebox{\textwidth}{!}{
\begin{tabular}{llll}%
\hline 
Algorithm&ACC&FNR&FPR\\\hline
LBRM&0.9999 (0.9999,0.9999)&0.0001 (0.0001,0.0001)&0.0002 (0.0002,0.0002)\\
DBRM{\_}G{\_}PAR&0.9998 (0.9986,1.0000)&0.0002 (0.0001,0.0021)&0.0000 (0.0000,0.0000)\\
DBRM{\_}R{\_}PAR&0.9998 (0.9964,0.9999)&0.0002 (0.0001,0.0052)&0.0000 (0.0000,0.0000)\\
%DBRM{\_}G{\_}3&0.9998 (0.9959,1.0000)&0.0002 (0.0001,0.0056)&0.0002 (0.0000,0.0042)\\
DBRM{\_}R&0.9998 (0.9946,1.0000)&0.0002 (0.0001,0.0076)&0.0002 (0.0000,0.0056)\\
%DBRM{\_}R{\_}3&0.9998 (0.9953,1.0000)&0.0002 (0.0001,0.0068)&0.0002 (0.0000,0.0070)\\
%DBRM{\_}R{\_}4&0.9998 (0.9945,1.0000)&0.0002 (0.0001,0.0080)&0.0002 (0.0000,0.0069)\\
DBRM{\_}G&0.9998 (0.9942,1.0000)&0.0002 (0.0001,0.0082)&0.0002 (0.0000,0.0072)\\
%DBRM{\_}G{\_}2&0.9998 (0.9933,1.0000)&0.0002 (0.0001,0.0089)&0.0002 (0.0000,0.0048)\\
%DBRM{\_}G{\_}4&0.9998 (0.9932,0.9999)&0.0002 (0.0001,0.0097)&0.0002 (0.0000,0.0042)\\
%DBRM{\_}R{\_}2&0.9998 (0.9925,1.0000)&0.0002 (0.0001,0.0105)&0.0002 (0.0000,0.0032)\\
LASSO&0.9991 (-,-)&0.0013 (-,-)&{0.0000} (-,-)\\
RIDGE&0.9982 (-,-)&0.0026 (-,-)&0.0000 (-,-)\\
LXGBOOST&0.9980 (0.9980,0.9980)&0.0029 (0.0029,0.0029)&0.0000 (0.0000,0.0000)\\
LR&0.9963 (-,-)&0.0054 (-,-)&0.0000 (-,-)\\
DEEPNETS&0.9728 (0.8979,0.9979)&0.0384 (0.0018,0.1305)&0.0000 (0.0000,0.0153)\\
TXGBOOST&0.8283 (0.8283,0.8283)&0.0005 (0.0005,0.0005)&0.3488 (0.3488,0.3488)\\
RFOREST&0.8150 (0.6761,0.9991)&0.1972 (0.0003,0.3225)&0.0162 (0.0000,0.3557)\\
NBAYES&0.6471 (-,-)&0.0471 (-,-)&0.4996 (-,-)\\
\hline
\end{tabular}}
%\\[1pt]
}
\end{table}
 The training sample consisted of $n=64$ objects (32 of which are PHO, whilst the other 32 are not) and the test sample of the remaining $n_p=20702$ objects. The results of Table~\ref{t1} show that even with such a small training set most methods tend to perform very well. The naive Bayes classifier has the smallest accuracy with a huge number of false positives. The tree based methods also have comparably small accuracy, where tree based gradient boosting in addition delivers too many false positives. Random forests tend to have on average too many false negatives, though there is huge variation of performance between different runs ranging from almost perfect accuracy down to accuracy as low as the naive Bayes classifier.   
 
 The DBRM methods are among the best methods for this data set and there is practically no difference between DBRM\_R and DBRM\_G. The best median performance has LBRM which indicates that non-linear structures do not play a big role in this example and all the other algorithms based on linear features (LASSO, RIDGE, logistic regression, linear gradient boosting)  performed similarly well. LBRM gives the same result in all simulation runs, the parallel versions of DBRM give almost the same model as LBRM and only rarely add some non-linear features, whereas the single threaded versions of DBRM much more often include non-linear features (Table \ref{Tab:complexity}). The slight variation between simulation runs  suggests that in spite of the general good performance of DBRM\_G and DBRM\_R both algorithms have not fully converged in some runs.

 %However, it might also be the case that the training data sample was too small to be able to consistently detect significantly important non-linearities.   Only one non-linear feature, $atan(eccentricity)$,  appears to be of some importance and has been reported in more than 10 DBRM runs. 
 
 %Other non-linear features found in the 100 runs  are reported in excel spreadsheets available as supplementary material.

\refstepcounter{Example} \label{Ex:BreastCancer}
\subsubsection*{Example \arabic{Example}: Breast cancer classification}

%Example \ref{Ex:BreastCancer}

The second example consists of breast cancer data with observations from 357 benign and 212 malignant tissues. Features are obtained from digitized images of fine needle aspirates (FNA) of breast mass.
Ten real-valued features are computed for each cell nucleus: \emph{radius, texture, perimeter, area, smoothness, compactness, concavity, concave points, symmetry} and \emph{fractal dimension}. For each feature, the mean, standard error, and "worst" or largest value (mean of the three largest values) per image were computed, resulting in 30 input variables per image, see~~\citet{breastdata} for more details on how the features were obtained. A randomly selected quarter of the images was used as a training data set, the remaining images as a test set.

% \begin{table}[h]{
% \resizebox{\textwidth}{!}{
% \begin{tabular}{lccccccccc}%
% \hline 
% Algorithm&min.p&med.p&max.p&min.fn&med.fn&max.fn&min.fp&med.fp&max.fp\\ \hline 
% RIDGE&\textbf{0.9742}&\textbf{0.9742}&0.9742&0.0592&0.0592&0.0592&0.0037&\textbf{0.0037}&\textbf{0.0037}\\
% GMJMCMC&0.9437&0.9695&\textbf{0.9812}&0.0479&0.0536&0.1067&\textbf{0.0000}&0.0148&0.0361\\
% DEEPNETS&0.9225&0.9695&0.9789&\textbf{0.0305}&0.0674&0.1167&0.0000&0.0074&0.0949\\
% RGMJMCMC&0.9554&0.9683&0.9789&0.0479&0.0536&0.0809&0.0037&0.0148&0.0361\\
% NAIVEBAYESS&0.9671&0.9671&0.9671&0.0479&0.0479&0.0479&0.0220&0.0220&0.0220\\
% MJMCMC&0.9624&0.9624&0.9624&0.0756&0.0756&0.0756&0.0111&0.0111&0.0111\\
% LASSO&0.9577&0.9577&0.9577&0.0756&0.0756&0.0756&0.0184&0.0184&0.0184\\
% LXGBOOST&0.9554&0.9554&0.9554&0.0809&0.0809&0.0809&0.0184&0.0184&0.0184\\
% TXGBOOST&0.9484&0.9531&0.9601&0.0536&0.0647&0.0756&0.0291&0.0326&0.0361\\
% RFOREST&0.9038&0.9343&0.9624&0.0422&0.0914&0.1675&0.0000&0.0361&0.1010\\
% LR&0.9272&0.9272&0.9272&0.0305&\textbf{0.0305}&\textbf{0.0305}&0.0887&0.0887&0.0887\\ \hline 
% \end{tabular}}
% \\[1pt]
% \caption{Comparison of performance (Precision, FPR, FNR) of different algorithms for breast cancer data}\label{t2}
% }

% \end{table}

\begin{table}[tbh]{
\caption{\label{t2}Comparison of performance  (ACC, FPR, FNR) of different algorithms for breast cancer data. See caption of Table~\ref{t1} for details.}   
\resizebox{\textwidth}{!}{
\begin{tabular}{llll}%
\hline 
Algorithm&ACC&FNR&FPR\\\hline
DBRM{\_}R{\_}PAR&{0.9765} ({0.9695},{0.9812})&0.0479 (0.0479,0.0479)&0.0074 ({0.0000},0.0184)\\
DBRM{\_}G{\_}PAR&0.9742 (0.9695,0.9812)&0.0479 (0.0479,0.0536)&0.0111 (0.0000,0.0184)\\
RIDGE&{0.9742} (-,-)&0.0592 (-,-)&{0.0037} (-,-)\\
LBRM&0.9718 (0.9648,0.9765)&0.0592 (0.0536,0.0702)&0.0074 (0.0000,0.0148)\\
DBRM{\_}G&0.9695 (0.9554,0.9789)&0.0536 (0.0479,0.0809)&0.0148 (0.0037,0.0326)\\
%DBRM{\_}G{\_}3&0.9695 (0.9507,0.9789)&0.0536 (0.0479,0.0862)&0.0148 (0.0000,0.0361)\\
%DBRM{\_}R{\_}2&0.9695 (0.9554,0.9789)&0.0536 (0.0422,0.0756)&0.0148 (0.0000,0.0396)\\
%DBRM{\_}R{\_}4&0.9695 (0.9577,0.9789)&0.0536 (0.0479,0.0756)&0.0148 (0.0000,0.0361)\\
DEEPNETS&0.9695 (0.9225,0.9789)&0.0674 ({0.0305},0.1167)&0.0074 (0.0000,0.0949)\\
DBRM{\_}R&0.9671 (0.9577,0.9812)&0.0536 (0.0479,0.0702)&0.0148 (0.0000,0.0361)\\
%DBRM{\_}R{\_}3&0.9671 (0.9577,0.9789)&0.0536 (0.0422,0.0756)&0.0148 (0.0037,0.0361)\\
%DBRM{\_}G{\_}4&0.9671 (0.9577,0.9789)&0.0536 (0.0305,0.0756)&0.0184 (0.0000,0.0361)\\
%DBRM{\_}G{\_}2&0.9671 (0.9531,0.9789)&0.0536 (0.0422,0.0862)&0.0184 (0.0000,0.0361)\\
LR&0.9671 (-,-)&0.0479 (-,-)&0.0220 (-,-)\\
LASSO&0.9577 (-,-)&0.0756 (-,-)&0.0184 (-,-)\\
LXGBOOST&0.9554 (0.9554,0.9554)&0.0809 (0.0809,0.0809)&0.0184 (0.0184,0.0184)\\
TXGBOOST&0.9531 (0.9484,0.9601)&0.0647 (0.0536,0.0756)&0.0326 (0.0291,0.0361)\\
RFOREST&0.9343 (0.9038,0.9624)&0.0914 (0.0422,0.1675)&0.0361 (0.0000,0.1010)\\
NBAYES&0.9272 (-,-)&{0.0305} (-,-)&0.0887 (-,-)\\ \hline
\end{tabular}}
} 
\end{table}

Qualitatively the results  presented in Table~\ref{t2} are quite similar to those from Example \ref{Ex:NeoAsteroids}. The naive Bayes classifier and random forests have the worst performance where NBAYES gives too many false positives and RFOREST too many false negatives, though less dramatically than in the previous example. All the algorithms based on linear features are performing much better which once again indicates that in this dataset non-linearities are not of primary importance. Nevertheless both versions of the DBRM algorithm, and in this example also DEEPNETS, are among the best performing algorithms. DBRM run on 32 parallel threads gives the highest median accuracy and performs substantially better than DBRM run only on one thread.  % Excel sheets with the non-linear features selected by DBRM are again reported as supplementary material.

\refstepcounter{Example} \label{Ex:Spam}
\subsubsection*{Example \arabic{Example}: Spam classification}
In this example we are using the data from \citet{cranor1998spam} for detecting spam emails. The concept of "spam" is extremely diverse and includes advertisements for products and web sites, money making schemes, chain letters, the spread of unethical photos and videos, et cetera. In this data set the collection of spam emails consists of  messages which have been actively marked as spam by users, whereas non-spam emails consist of messages filed as work-related or personal. 
%and results in highly specific non-spam indicators like the name 'George' or the area code '650'. 
The data set includes 4601 e-mails, with 1813 labeled as spam. For each e-mail, 58 characteristics are listed which can serve as explanatory input  variables. These include 57 continuous and 1 nominal variable, where most of these are concerned with the frequency of particular words or characters.  Three variables provide different measurements on the sequence length of consecutive capital letters.  The data was randomly divided into a training data set of 1536 e-mails and a test data set of the remaining 3065 e-mails.

%The data-set thus emulates the construction of a personalized spam filter.
%To generate a general purpose spam filter one would either have to blind such non-spam indicators or get a much wider collection of e-mails classified as non-spam.
%Other variables include 48 continuous real in the range [0,100] attributes of percentage of words in the e-mail that match WORD, i.e. 100 * (number of times the WORD appears in the e-mail) / total number of words in e-mail.  A "word" in this case is any string of alphanumeric characters bounded by non-alphanumeric characters or end-of-string.
%The data also contains 6 continuous real [0,100] variables giving the percentage of characters in the e-mail that match CHAR, i.e. 100 * (number of CHAR occurences) / total characters in e-mail. 
%Another variable is associated with average length of uninterrupted sequences of capital letters. One natural valued attribute describing the length of longest uninterrupted sequence of capital letters and one natural valued variable describing the total number of capital letters in the e-mail are also present. 

 Table~\ref{t3} reports the results for the different methods. Once again the naive Bayes classifier performed worst. Apart from that the order of performance of the algorithms is quite different from the first two examples. The tree based algorithms show the highest accuracy whereas the five algorithms based on linear features have less accuracy. This indicates that non-linear features are important in this dataset to discriminate between spam and non-spam. As a consequence DBRM performs  better than LBRM. 
 
Specifically the parallel version of DBRM provides almost the same accuracy as DEEPNETS, with the minimum accuracy over 100 runs being actually larger, the median and maximum accuracy quite comparable. However, tree based gradient boosting and random forests perform substantially better which is mainly due to the fact that they can optimize cutoff points for the continuous variables. One way to potentially improve the performance of DBRM would be to include multiple characteristic functions, like for example  $\text{I}(x>\mu_{x}), \text{I}(x<F^{-1}_{0.25}(x)), \text{I}(x>F^{-1}_{0.75}(x))$, into the set of non-linear transformations $\cal{G}$ to allow the generation of features with splitting points like in random trees.
%allow more splitting points in random trees.
%  \todo{ AH: But what we can do is to add more splitting functions, e.g. $I(x>mean(x)), I(x>q(x,025)), I (x>q(x,0.75))$
%  \\ FF: I would probably only add this as a comment.} 
 %Notice that for this example performance of non-linear approaches is significantly better than performance of the linear approaches, indicating potential non-linearities in the features. 

\begin{table}[tb]{
\caption{Comparison of performance (ACC, FPR, FNR) of different algorithms for spam data. For methods with random outcome the median measures (with minimum and maximum in parentheses) are displayed.
The algorithms are sorted according to median power.}\label{t3}
\resizebox{\textwidth}{!}{
\begin{tabular}{llll}%
\hline 
Algorithm&ACC&FNR&FPR\\\hline
TXGBOOST&0.9465 (0.9442,0.9481)&0.0783 (0.0745,0.0821)&0.0320 (0.0294,0.0350)\\
RFOREST&0.9328 (0.9210,0.9413)& 0.0814 (0.0573,0.1174)&0.0484 (0.0299,0.0825)\\
DEEPNETS&0.9292 (0.9002,0.9357)& 0.0846 (0.0573,0.1465)&0.0531 (0.0310,0.0829)\\
DBRM{\_}R{\_}PAR&0.9268 (0.9162,0.9390)&0.0897 (0.0780,0.1057)&0.0538 (0.0415,0.0691)\\
DBRM{\_}G{\_}PAR&0.9251 (0.9139,0.9377)&0.0897 (0.0766,0.1024)&0.0552 (0.0445,0.0639)\\
%DBRM{\_}G{\_}2&0.9243 (0.9100,0.9357)&0.0927 (0.0780,0.1103)&0.0545 (0.0445,0.0686)\\
DBRM{\_}G&0.9243 (0.9113,0.9328)&0.0927 (0.0808,0.1116)&0.0552 (0.0465,0.0658)\\
%DBRM{\_}G{\_}3&0.9237 (0.9100,0.9321)&0.0924 (0.0766,0.1122)&0.0548 (0.0474,0.0714)\\
%DBRM{\_}G{\_}4&0.9237 (0.9113,0.9315)&0.0931 (0.0821,0.1077)&0.0562 (0.0470,0.0714)\\
DBRM{\_}R&0.9237 (0.9106,0.9351)&0.0917 (0.0801,0.1116)&0.0557 (0.0474,0.0672)\\
%DBRM{\_}R{\_}4&0.9225 (0.9119,0.9315)&0.0958 (0.0780,0.1116)&0.0550 (0.0450,0.0667)\\
%DBRM{\_}R{\_}3&0.9207 (0.6509,0.9299)&0.0994 (0.0372,0.1231)&0.0545 (0.0440,0.3564)\\
%DBRM{\_}R{\_}2&0.9197 (0.9054,0.9282)&0.1014 (0.0835,0.1168)&0.0545 (0.0479,0.0691)\\
LR&0.9194 (-,-)&0.0681 (-,-)&0.0788 (-,-)\\
LBRM&0.9178 (0.9168,0.9188)&0.1090 (0.1064,0.1103)&0.0528 (0.0523,0.0538))\\
LASSO&0.9171 (-,-)& 0.1077 (-,-)&0.0548  (-,-)\\
RIDGE&0.9152 (-,-)&0.1288 (-,-)& 0.0415 (-,-)\\
LXGBOOST&0.9139 (0.9139,0.9139)&0.1083 (0.1083,0.1083)&0.0591 (0.0591,0.0591)\\
NBAYES&0.7811 (-,-)&0.0801 (-,-)&0.2342 (-,-)\\
\hline
\end{tabular}}
}
\end{table}

\subsubsection*{Complexities of the features for the prediction examples}
One can conclude from these three examples that DBRM has good predictive performance both when non-linear patterns are present (Example \ref{Ex:NeoAsteroids} and \ref{Ex:BreastCancer}) or when they are not (Example \ref{Ex:Spam}). Additionally DBRM has the advantage that its generated features are highly interpretable.  Excel sheets are provided as supplementary material which present all features detected by DBRM with posterior probability larger than 0.1 and Table \ref{Tab:complexity} provides the corresponding frequency distribution of the complexity of these features.    

In Example \ref{Ex:NeoAsteroids} concerned with the asteroid data all reported non-linear features had a complexity of 2. As mentioned previously the parallel version of DBRM detected way less non-linear features than the simple versions which suggests that DBRM\_G and DBRM\_R have not completely converged in some simulation runs. Approximately half of the non-linear features were modifications and the other half interactions. In this example not a single projection was reported in all simulation runs by any of the DBRM implementations.

\begin{table}[]{
\caption{Mean frequency distribution of feature complexities detected by the different DBRM algorithms in 100 simulation runs for the first three examples. The final row for each example gives the mean of total number of features in 100 simulation runs which had a posterior probability larger than 0.1.}\label{Tab:complexity}
\resizebox{\textwidth}{!}{
\begin{tabular}{llllll}%
\multicolumn{6}{l}{\underline{\textbf{Example1}: Asteroid}}\\
%\hline
complexity&DBRM\_G&DBRM\_R&DBRM\_G\_PAR&DBRM\_R\_PAR&LBRM\\
1&8.9600&8.9700&9.0000&9.0000&9.0000\\
2&2.5800&2.6200&0.0500&0.1500&0.0000\\
\hline
Total&11.540&11.590&9.0500&9.1500&9.0000
\\[1mm]
\multicolumn{6}{l}{\underline{\textbf{Example2}: Breast cancer}}\\
%\hline
complexity&DBRM\_G&DBRM\_R&DBRM\_G\_PAR&DBRM\_R\_PAR&LBRM\\
1&11.300&11.730&14.200&10.790&29.830\\
2&3.0900&3.0600&0.0400&0.2100&0.0000\\
3&0.3000&0.0000&0.0000&0.0000&0.0000\\
6&0.0000&0.0100&0.0000&0.0000&0.0000\\
7&0.0000&0.0100&0.0000&0.0000&0.0000\\
\hline
Total&14.420&14.810&14.240&11.000&29.830\\[1mm]
\multicolumn{6}{l}{\underline{\textbf{Example3}: Spam mail}}\\
%\hline
complexity&DBRM\_G&DBRM\_R&DBRM\_G\_PAR&DBRM\_R\_PAR&LBRM\\
1&36.340&36.090&39.870&39.170&49.830\\
2&14.450&14.830&21.470&22.430&0.0000\\
3&2.8300&3.1700&5.2400&5.8100&0.0000\\
4&0.6900&0.5700&1.3600&1.3600&0.0000\\
5&1.1500&1.0900&1.5600&1.6800&0.0000\\
6&0.9200&0.7400&1.2400&1.0700&0.0000\\
7&0.3700&0.4000&0.5700&0.4200&0.0000\\
8&0.2500&0.2200&0.3300&0.1700&0.0000\\
9&0.0400&0.0800&0.1600&0.1100&0.0000\\
$\geq$10&0.1500&0.1100&0.1100&0.1800&0.0000\\
\hline
Total&57.190&57.300&71.910&72.400&49.830\\\hline
\end{tabular}}
}
\end{table}

Also in Example \ref{Ex:BreastCancer} the parallel versions of DBRM reported a substantially smaller number of non-linearities than the single-threaded version. Over all simulation runs only DBRM\_R detected 2 projections (with complexity 6 and 7, respectively). Otherwise in this example interactions were more often detected than modifications. Interestingly the non-linear features reported by the parallel versions of DBRM consisted  only of the following two interactions:  (standard error of the area) $\times$ (worst texture) reported 3 times by DBRM\_G\_PAR  and 10 times by DBRM\_R\_PAR and (worst texture) $\times$ (worst concave points) reported once by DBRM\_G\_PAR  and 11 times by DBRM\_R\_PAR. While LBRM includes almost always all 30 variables in the model (in 100 simulation runs only 17 out of 3000 possible linear features had posterior probability smaller than 0.1), DBRM delivers more parsimonious models. %This is apparently the price one has to pay for searching over a much larger model space. 

In Example \ref{Ex:Spam} there is much more evidence for non-linear structures. The non-linear features with highest detection frequency over simulation runs in this example were always modifications. For DBRM\_R\_PAR there were 10 modifications of depth 2  which were detected in more than 25 simulation runs. For example sin($X_{7}$) was reported 46 times and gauss($X_{36}$) 41 times. The features atan($X_{52}$) and  tanh($X_{52}$) were reported 41 times and 38 times, respectively, which provides strong evidence that a non-linear transformation of $X_{52}$ is an important predictor. For DBRM\_G\_PAR the results are quite similar and the mentioned four modifications are also among the top-ranking non-linear features. Although modifications were most important in terms of replicability over simulation runs, in this example DBRM also found many interactions and projections. 
From the 3204 non-linear features reported by DBRM\_G\_PAR there were 
more than 998 which included one interaction, 116 with two interactions and even 3 features with three interactions. Furthermore there were 353 features including one projection, 12 features with two nested projections and even 3 features where three projections were nested. However, these highly complex features typically occurred only in one or two simulation runs. In spite of the really good performance of the parallel versions of DBMR it seems that even more parallel threads and longer chains might be necessary to get consistent results over simulation runs in this example. 

\subsection{Model inference}

Examples \ref{Ex:JupiterMass} and \ref{Ex:Kepler} are based on data sets describing physical parameters of newly discovered exoplanets. The data was originally collected and continues to be updated by Hanno Rein at the Open Exoplanet Catalogue Github repository \citep{exocat}. The input covariates include planet and host star attributes, discovery methods, and dates of discovery.  We use a subset of $n = 223$ samples containing all planets with no missing values  to rediscover two basic physical laws which involve some non-linearities. 
We compare the performance of DBRM\_G and DBRM\_R when running different numbers of parallel threads. We restrict ourselves to DBRM here because to our best knowledge no other machine learning approaches can be used for the detection of sophisticated non-linear relationships in closed form.

For both examples we utilize DBRM models with conditionally independent Gaussian observations:
\begin{align}
  Y_i|\mu_i,\sigma^2 \sim&  N(\mu_i,\sigma^2),\quad i = 1,...,n\label{exomodel0}\\ 
  \mu_i =& \beta_0  + \sum_{j=1}^{q} \gamma_{j}\beta_{j}F_{j}(\boldsymbol{x}_{i})\label{exomodel1} \; .
\end{align}
We consider two different sets of non-linear transformations, $\mathcal{G}_1 = \{$$\text{sigmoid}(x)$, $\sin(x)$, $\tanh(x)$, $\text{atan}(x)$, $|x|^\frac{1}{3}$$\}$ and $\mathcal{G}_2 = \{\text{sigmoid}(x), \sin(x),\exp(-|x|),\log(|x|+1), |x|^\frac{1}{3}$, $|x|^{2.3}$, $|x|^{3.5}\}$, where we 
 restrict the depth to $D_{max} = 5$ and the maximum number of features in a model to $Q=15$. $\mathcal{G}_1$ is an adaptation of the set of transformations used in the first three examples. Adding $|x|^\frac{1}{3}$ results in a model space which includes a closed form expression of Kepler's 3rd law in Example \ref{Ex:Kepler}.
 $\mathcal{G}_2$ is a somewhat larger set where the last two functions are specifically motivated to facilitate generation of interesting features linking the mass and luminosity of stars \citep{kuiper1938empirical,salaris2005evolution}.  
 
 For the prior of the model structure~\eqref{eq:modelprior} we choose $a=e^{-2 \log n}$ giving a BIC like penalty for the model complexity. The
parameter priors are specified as  
\begin{align}
%p(\boldsymbol\gamma)\propto &  \prod_{j=1}^{q}\exp\left(-2\log n\gamma_jc(F_j(\B x))  \right)\; , \label{gammodel1} \\
 \pi(\sigma^2) =& \sigma^{-2} \\ p(\boldsymbol{\beta}|\boldsymbol{\gamma},\sigma^2)=&|J^{\boldsymbol{\gamma}}_n(\boldsymbol{\beta},\sigma^2)|^{\frac{1}{2}},\;\label{JefPrior}
\end{align}
where $|J^{\boldsymbol{\gamma}}_n(\boldsymbol{\beta},\sigma^2)|$ is the determinant of the corresponding Fisher information matrix. Hence \eqref{JefPrior} is Jeffrey's prior for the coefficients. 
In this case, marginal likelihoods can be computed exactly.
%Similarly as in the first three examples on prediction this choice of priors allows us to use Laplace's approximation to compute marginal likelihoods \citep{Hubin2017}. 

 The focus in these examples is on correctly identifying important features. Consequently we are using a  threshold value of $\eta^* = 0.25$ for the feature posteriors to define positive detections which is larger than the threshold $0.1$ used when reporting relevant features for prediction in the first three examples. To evaluate the performance of algorithms  we  report estimates for the power (Power), the false discovery rate (FDR), and the expected number of false positives (FP) based on $N$ simulation runs.  These measures are defined as follows.  
\begin{align*}
\text{Power} = &
N^{-1}\sum_{i=1}^N \text{I}(\hat\gamma^i_{j^*}=1);\\
%\text{FDR} 
%=&\frac{\sum_{i=1}^N  \text{I}(\hat\gamma^i_j=1,\gamma_j = 0)}{\sum_j\text{I}(\hat\gamma_j^i=1)};\\
%=&\sum_{i=1}^N
%    \frac{\text{TP}_i+\text{FP}_i}
%         {\sum_{i'=1}^N( \text{TP}_{i'}+\text{FP}_{i'})}
%   \frac{ \text{FP}_i}{\text{TP}_i+ \text{FP}_i};&&\text{Not to be included in the end}\\
\text{FDR} =&N^{-1}\sum_{i=1}^N\frac{\sum_j\text{I}(\gamma_j=0,\hat\gamma_j^i=1)}{\sum_j\text{I}(\hat\gamma_j^i=1)}\\  %\text{FP}_i=&\sum_{j}I(\gamma_j=0,\hat\gamma_j^i=1)\\
%\text{TP}_i=&\sum_{j}I(\gamma_j=1,\hat\gamma_j^i=1)\\
\text{FP} =& N^{-1}\sum_{i=1}^N  \sum_{j\neq j^*}\text{I}(\hat\gamma^i_j=1).
\end{align*}
Here $\hat\gamma^i_j = \text{I}(\hat p(\gamma_j|\By)>\eta^*)$ denotes the identification of $\gamma_j$ in run $i$ of the algorithm and $j^*$ is the index of a true feature, which means a feature which is in accordance with the well known physical laws.  For Kepler's third law several features can be seen as equivalent true positives and consequently the definition of Power and FDR will be slightly modified.

\refstepcounter{Example} \label{Ex:JupiterMass}
\subsubsection*{Example \arabic{Example}: Jupiter mass of the planet}

In this example we consider the planetary mass as a function of its radius and density. It is common in astronomy to use the measures of Jupiter as units and a basic physical law gives the non-linear relation 
\begin{align} \label{mass_law}
	m_p\approx R^3_p\times \rho_p \;.
\end{align}
Here $m_p$ is  the planetary mass $m_p$ measured in units of Jupiter mass (denoted \textit{PlanetaryMassJpt} from now on). Similarly the radius of the planet $R_p$  is measured in units of Jupiter radius  and the density of the planet $\rho_p$  is measured in units of Jupiter density. Hence in the data set the variable  \textit{RadiusJpt} refers to $R_p$, and \textit{PlanetaryDensJpt} denotes $\rho_p$. The approximation sign is used because the planets are not exactly spherical but rather almost spherical ellipsoids.

DBRM according to \eqref{exomodel0}-\eqref{exomodel1} is used to model \textit{PlanetaryMassJpt} as a function of the following ten potential input variables: \textit{TypeFlag}, \textit{RadiusJpt}, \textit{PeriodDays}, \textit{SemiMajorAxisAU}, \textit{Eccentricity}, \textit{HostStarMassSlrMass}, \textit{HostStarRadiusSlrRad}, \textit{HostStarMetallicity}, \textit{HostStarTempK}, \textit{PlanetaryDensJpt}. 
In order to evaluate the capability of GMJMCMC and RGMJMCMC to detect true signals we run each algorithm $N = 100$ times. 
To illustrate to which extent the performance of DBRM depends on the number of parallel runs we furthermore consider computations with 1, 4 and 16 threads, respectively. 
In each of the threads the algorithms were first run for 10\,000 iterations, with population changes at every 250 iteration, and then for a larger number of iterations based on the last population (until a total number of 10\,000 unique models was obtained). Results for GMJMCMC and RGMJMCMC using different numbers of threads are summarized in Table \ref{Tab:Results_Ex2}  both for $\mathcal{G}_1$ and $\mathcal{G}_2$.

\begin{table}[tb]
\centering
\caption{Power, False Positives (FP) and FDR for detecting the mass law \eqref{mass_law} based on the decision rule that the posterior probability of a feature  is larger than $\eta^* = 0.25$. The feature $R\times R\times R\times \rho_p$  is counted as true positive, all other selected features as false positive. DBRM is applied using the non-linear sets (NL set) $\mathcal{G}_1$ and $\mathcal{G}_2$ and different numbers of parallel threads.} 
\label{Tab:Results_Ex2}
\begin{tabular}{|cl|ccc|ccc|}%ccc|}
\hline
&&\multicolumn{3}{c|}{DBRM{\_}G{\_}PAR} &\multicolumn{3}{c|}{DBRM{\_}R{\_}PAR}\\\hline %& \multicolumn{3}{c|}{RGMJMCMC(NA)}\\ \hline
NL set&Threads&Power&FP&FDR&Power&FP&FDR\\\hline%&Power&FP&FDR\\\hline
$\mathcal{G}_1$&16&1.00 &0.00&0.00&0.97&0.06&0.03\\%0.0&2.22&1.0\\
&4&0.79&0.40&0.21 &0.61&0.73&0.39\\%&&&\\
&1&0.42&1.21&0.58 & 0.33&1.63&0.67\\\hline%&&&\\\hline
%\hline
$\mathcal{G}_2$&16&0.93&0.36&0.215&0.94&0.29&0.175\\%0.0&2.22&1.0\\
&4&0.69&0.49&0.34 &0.63&0.64&0.375\\%&&&\\
&1&0.42&1.25&0.58&0.29&1.54&0.71\\\hline%&&&\\\hline
\end{tabular}
\end{table}

Clearly the more resources become available the better DBRM performs. 
RGMJMCMC and GMJMCMC both manage to find the correct model with rather large Power (reaching gradually one) and small FDR (reaching gradually zero), when the number of parallel threads is increased. When using only a single thread it often happens that instead of the correct feature some closely related features are selected (see the Excel sheet Mass.xlsx in the supplementary material for more details). Results for the set $\mathcal{G}_1$ are slightly better than for $\mathcal{G}_2$ which illustrates the importance of having a good set of transformations when interested in model inference. Power is lower and FDR is larger for $\mathcal{G}_2$ which is mainly due to the presence of $|x|^{3.5}$ in the set of nonlinearities. The feature $R^{3.5}_p\times \rho_p$ is quite similar to the correct law \eqref{mass_law}  and moreover has lower complexity than the feature $R\times R\times R\times \rho_p$. Hence it is not surprising that it is often selected, specifically when DBRM was not run sufficiently long to fully explore features with larger complexity.

\refstepcounter{Example} \label{Ex:Kepler}
\subsubsection*{Example \arabic{Example}: Kepler's third law}
In this example we want to model the semi-major axis of the orbit  $a=$\textit{SemiMajorAxisAU} as a function of the following 10 potential input variables: \textit{TypeFlag}, \textit{RadiusJpt}, \textit{PeriodDays}, \textit{PlanetaryMassJpt}, \textit{Eccentricity}, \textit{HostStarMassSlrMass}, \textit{HostStarRadiusSlrRad}, \textit{HostStarMetallicity}, \textit{HostStarTempK}, \textit{PlanetaryDensJpt}. 

Kepler's third law says that the square of the orbital period $P$ of a planet is directly proportional to the cube of the semi-major axis $a$ of its orbit. Mathematically this can be expressed as
\begin{align}
\frac{P^2}{a^3}=\frac{4\pi^2}{G(M+m)}\approx\frac{4\pi^2}{GM},\label{3kepl0}
\end{align}
where $G$ is the gravitational constant, $m$ is the mass of the planet, $M$ is the mass of the corresponding hosting star and $M\gg m$. 
Equation \eqref{3kepl0} can be reformulated as
\begin{align*}
a \approx K\left(P^2 M_h \right)^{\frac{1}{3}},
\end{align*}
where the  approximation is due to neglecting $m$. Here the mass of the hosing star $M_h$ is measured in the unit of Solar mass and thus the constant $K$ includes not only the gravitational constant $G$ but also the normalizing constant for the mass.  There exist certain power laws which relate the mass $M_h$ of a star with its radius $R_h$  as well as with its temperature $T_h$. Although these relationships are not linear it is still not particularly surprising that there are two features which are strongly correlated with the target feature, namely $(R_h P^2)^{\frac{1}{3}}$ (with a correlation of 0.9999667) and $(T_h P^2)^{\frac{1}{3}}$ (with a correlation of 0.9995362).

In order to assess the ability of GMJMCMC and RGMJMCMC to detect these  features we performed again $N=100$ runs for both $\mathcal{G}_1$ and $\mathcal{G}_2$ when using 1, 16, and 64 threads, respectively. The number of iterations in each thread was defined exactly like in Example \ref{Ex:Kepler} to obtain 20\,000 unique models after the last swap of  populations.  The results for GMJMCMC and RGMJMCMC are  presented in Table~\ref{Tab:Results_Kepler}. A detection of any of the three highly correlated features described above is counted as a true positive, other features are counted as false positives, and the definitions of Power and FDR are modified accordingly.

\begin{table}[tb]
\centering
\caption{Results for detecting Kepler's third law \eqref{3kepl0} based on the decision rule that the posterior probability of a feature  is larger than $\eta^* = 0.25$. The three features  $\left(P\times P \times  M_h\right)^{\frac{1}{3}}$,  $\left(P\times P \times  R_h \right)^{\frac{1}{3}}$  and  $\left(P \times P \times T_h \right)^{\frac{1}{3}}$ are counted as true positives, all other selected features as false positives. Apart from the power to detect each of these features ($F_1, F_2$ and $F_3$) we report the power to detect at least one of them (Pow), the number of other detected features (FP) and the corresponding false discovery rate (FDR).  DBRM is applied using the non-linear sets (NL set) $\mathcal{G}_1$ and $\mathcal{G}_2$ and different numbers of parallel threads.} 
\label{Tab:Results_Kepler}
\begin{tabular}{|cl|cccccc|cccccc|}%ccc|}
\hline
&&\multicolumn{6}{c|}{DBRM{\_}G{\_}PAR} &\multicolumn{6}{c|}{DBRM{\_}R{\_}PAR}\\\hline %& \multicolumn{3}{c|}{RGMJMCMC(NA)}\\ \hline
NL set&Threads& $F_1$& $F_2$& $F_3$ &Pow& FP&FDR&
  $F_1$& $F_2$& $F_3$ &Pow& FP&FDR\\\hline
$\mathcal{G}_1$ &64& 81& 71&1 &1.00& 0.02&0.01& 78&75&2 &0.99 &0.03  &0.01 \\
&16& 34& 41&32&0.84&0.46&0.18 & 31&38&18&0.79 &0.68 &0.25 \\
& 1&    6& 5 &3 &0.141&0.65&0.86& 6 & 4& 2&0.12 &1.81&0.88 \\
\hline %\hline
$\mathcal{G}_2$ &64& 72& 71&3 &0.99&0.04&0.015& 70&68&9 &1.00&0.04&0.02 \\
&16& 39& 42&13&0.83&0.55&0.22 & 24&27&16&0.65&0.88&0.39 \\
&1&    7& 4 &3 &0.14&1.81&0.86& 2 & 2& 2&0.06&2.14&0.94 \\\hline
\end{tabular}
\end{table}

Qualitatively the results are similar to Example \ref{Ex:JupiterMass}.
 With increasing computational effort  Power is converging to 1 and FDR is getting close to 0 both for GMJMCMC and RGMJMCMC. On average GMJMCMC is performing slightly better than RGMJMCMC. In this example there is not such a big difference between the non-linear sets $\mathcal{G}_1$ and $\mathcal{G}_2$.
 For both examples these results were obtained with a fairly small sample size of $n = 223$ observations. 
In Appendix~\ref{ap:interp} we discuss in more detail the importance of using flexible feature spaces to obtain interpretable models. The main conclusion is that when the set of non-linear transformations is too restricted, more complex features are required to explain the same non-linear relationships. 

%\subsection{Example 3: Inference with latent Gaussian variables. A bit-coin study.}

%In this section we will address the example on the inference on bit-coin rates using various market variables and %a number of latent Gaussian variables.

\section{{Summary and discussion}}\label{section5}
In this article we have introduced a new class of deep Bayesian regression models  (DBRM) to perform automated feature engineering in a Bayesian context. The approach is easily extended to include latent Gaussian variables to model different correlation structures between individuals. Two algorithms are introduced to estimate model posterior probabilities, the genetically modified MJMCMC approach (GMJMCMC) as well as its reversible modification (RGMJMCMC). These algorithms combine two key ideas, firstly having a population (or search space)  of highly predictive features which is regularly updated and secondly using MJMCMC to efficiently explore models including features within these populations. In the reversible version transitions between populations are constructed in such a way that detailed balance equation is satisfied throughout and hence the equilibrium distribution of RGMJMCMC can be used to estimate posterior probabilities. 

In several examples we have shown that the suggested approach can be efficient not only for prediction but also for model inference. In the prediction driven examples there is hardly any difference between the performance of GMJMCMC and RGMJMCMC, whereas GMJMCMC tends to perform slightly better in terms of inference. Inference within DBRM often requires significant computational resources, hence parallel runs of GMJMCMC (RGMJMCMC), and merging results in the end, is recommended. The resulting benefits have been illustrated in several examples. 
A memory efficient way of performing parallelized DBRM is implemented in the \textit{EMJMCMC} R-package which is currently available from the GitHub repository \citep{gmjref}. The developed package gives the user high flexibility both in the choice of methods to obtain marginal likelihoods and in prior specification.

One of  the main advantage of Bayesian deep learning is the possibility to quantify the uncertainty of  predictions. Currently, commonly used  Bayesian approaches to deep learning rely on variational Bayes approximations~\citep{Gal2016Uncertainty}, which tend to be rather crude. 
%He also gives some ideas on the crude approximations for the latter based on variational  for the models with Gaussian observations as well as some heuristics for the non Gaussian models. 
In contrast our approach provides well defined and mathematically justified uncertainty measures for any parameter $\Delta$ of interest, which can be naturally derived through standard Bayesian model averaging. This also allows for calculation of credibility intervals.
\begin{comment}
\begin{equation}\label{CIS1}
p(\Delta \in \Omega_\Delta|\By) =  \sum_{\M \in\MS}{p(\Delta \in \Omega_\Delta|\M,\By)p(\M|\By)} = 1 -\alpha.
\end{equation}
Resolving the equation above with respect to $\Omega_\Delta$ for any fixed $\alpha$ will give a credible region for $\Delta$. For example in the univariate case if we are interested in a one-sided credible interval $\Delta\leq\Delta_{ub}^\alpha$ with probability $1-\alpha$ we get 
\begin{equation}\label{CIS2}
\Delta_{ub}^\alpha = F^{-1}(1-\alpha|\By).
\end{equation}
\todo{$F$ is not defined here!}
which can be easily computed by model averaging.
\end{comment}

At the same time there are still several important questions  open for discussion. It is far from obvious how to optimize the choice of weights in the  feature engineering step. In this article we have used a computationally and assumption-pragmatic strategy, based on first estimating parameters on the outer layer of the feature and then taking a nonlinear modification of the obtained feature. However, we have implemented three further strategies, including optimization of weights from the last nonlinear projection, optimization with respect to all layers of a feature and a fully Bayesian search. The first two strategies are computationally more demanding than the default strategy and rely upon additional assumptions on the nonlinear transformations involved. The third one provides a fully Bayesian approach but is extremely slow in terms of convergence. A detailed description of these strategies is given in  Appendix~\ref{ap:add.ex1_3} of the supplementary material. We have run DBRM with these strategies for the first three Examples of Section \ref{section4} and the results are reported in Appendix~\ref{ap:add.ex1_3}. However, none of these strategies clearly outperforms the  simple strategy from Section \ref{section4}. Further research in this direction is necessary and should include simulation scenarios where nested projections are part of the data generating model. 

%In many cases, covariates from temporally or spatially close observations could be important to be included, similar to the recurrent NN structures.
%This can be achieved in DBRM by extending $\mathcal{G}$ to include temporal or spatial lag functions. Such an approach both allows for selection of the important lag features but at the same time allows for penalization of higher order lags through priors.
%One might discuss some additional functions to be included in set $\mathcal{G}$. 
%A pair of functions $\sin(kx)$ and $\cos(kx)$ for example would allow to have all Fourier series within the feature space of interest. 
%In addition, in the current notation the recurrent ANN structures are not exploited automatically unless the lags of interest (of both features and responses) are included manually in the set of input features.
%If one however extends the features to depend on all of the observations jointly, i.e. defining $F_j(\B x)$ instead of $F_j(\B x_i)$ in equation \eqref{DeepModel}, then with and $\text{lag}_p(x)$ operator, making a lag of the corresponding variable $x$, one is able to explore various recurrent structures in space in time automatically. The latter in combination with the latent Gaussian variables gives additional flexibility in modeling spatial-temporal relations. Whilst not to confuse the reader this more general notation is not introduced in \eqref{DeepModel},
%the approach is perfectly possible within our implementation if one uses one or several $\text{lag}_p(x)$ functions in $\mathcal{G}$.

An important issue left for discussion is how to manage large data samples (also known as Big Data) with the DBRM approach.
As for the marginal likelihood calculated with respect to parameters across all of the layers, only very crude approximate solutions based on the variational Bayes approach \citep{jordan1999introduction} are currently scalable for such problems~\citep{barber1998ensemble,blundell2015weight}. \citet{mackay1992practical,denker1991transforming} applied the Laplace approximations to approximate marginal likelihood across all layers. This approach is also very demanding computationally and can not be easily combined with combinatorial search of the best architectures in a time friendly way. \citet{neal2012bayesian} suggests a Hamiltonian Monte Carlo (HMC) to make proper Bayesian inference on Bayesian neural networks. Unfortunately his approach is even more computationally demanding and hence does not seem scalable to high dimensional model selection. To reduce computational complexity of HMC and improve its scalability to large data sets, \citet{welling2011bayesian} suggested to use stochastic estimates of the gradient of the likelihood. 
Many recent articles describe the possibility of such sub-sampling combined with MCMC  \citep{quiroz2014speeding, quiroz2017speeding, quiroz2016exact, flegal2012applicability, pillai2014ergodicity}, where unbiased likelihood estimates are obtained from subsamples of the whole data set in such a way, that ergodicity and the desired limiting properties of the MCMC algorithm are sustained. These methods are not  part of the current implementation of DBRM, but our approach can relatively easily be adopted to allow sub-sampling MCMC techniques in the future.

%\newpage
\bigskip
\begin{center}
{\large\bf SUPPLEMENTARY MATERIAL}
\end{center}
\textbf{R package:} \textit{R} package \textit{EMJMCMC} for (R)(G)MJMCMC \citep{gmjref}. (EMJMCMC\_ 1.4.tar.gz; GNU zipped tar file)\\
\textbf{Data and code:} Data (simulated and real) and \textit{R} code for (R)(G)MJMCMC algorithms(code-and-data.zip \citep{supt}; zip file containing the data, code and a read-me file (readme.pdf))\\
\textbf{Additional materials:} Additional tables and examples. (appendix.pdf)

\footnotesize

\let\oldbibliography\thebibliography
\renewcommand{\thebibliography}[1]{\oldbibliography{#1}
\setlength{\itemsep}{0pt}} %Reducing spacing in the bibliography.

\bigskip
\begin{center}
{\large\bf ACKNOWLEDGMENTS}
\end{center}

We thank CELS project (\url{http://www.mn.uio.no/math/english/research/groups/cels/}) at the University of Oslo for giving us the opportunity, inspiration and motivation to write this article. We would also like to acknowledge NORBIS for funding academic stay of the first author in Vienna (see \url{https://norbis.w.uib.no/an-autumn-with-bayesian-approaches-in-vienna/} for more detail).

\newpage
\bibliography{ref}
\bibliographystyle{agsm}

\newpage
\appendix

\section{Alternative strategies for specifying weights}\label{ap:add.ex1_3}

In section~\ref{Subsec:FeatureEngineering} one specific choice for specifying the weights $\BB\alpha$ in feature engineering was introduced where weights are obtained by optimizing \eqref{Strategy_1}. The corresponding strategy might be abbreviated as 'optimize then transform', because the non-linear transformation happens after the weights have been computed. Here we present three alternative strategies of increasing computational complexity.

\begin{description}
\item[Strategy 2 (transform then optimize):] 
Like in the original strategy the  weights $\BB\alpha$  are specified conditional on the $F_{r_l}(\boldsymbol{x})$ terms defined at earlier steps but now optimization happens after applying the transformation $g(\cdot)$. In other words the weights are obtained as  maximum likelihood estimates using  model~\eqref{themodeleq} with $F_{r_l}, r_l  = 1,...,w_j$ as covariates and $g^{-1}(\mathsf{h}(\cdot))$ as a link function, thus fitting the model $ g^{-1}(\mathsf{h}(\mu)) = \boldsymbol{\alpha}_j^T\B F^d(\boldsymbol{x})  + \alpha_{j,0} \; .$ This strategy yields a particularly simple optimization problem if $\mathsf{h}$ is the canonical link function and $g(\cdot)$ a concave function in which case the estimates are uniquely defined. However, if we want to use gradient based optimizers then we have to make a restriction on $g(\cdot)$ to be continuous and differentiable in the regions of interest. Otherwise gradient free continuous optimization techniques have to be applied. 

\item[Strategy 3 (transform then optimize across all layers):] 
%The parameters are specified jointly for all the $F_{r_l}(\boldsymbol{x})$ terms defined at earlier steps and the currently selected $g(\cdot)$. 
Similarly as in Strategy 2 parameters are obtained as  maximum likelihood estimates using  model~\eqref{themodeleq} but we include now parameters from all layers as covariates. % and $g^{-1}(\mathsf{h}(\cdot))$ as a link function. 
We are again fitting the model $g^{-1}(\mathsf{h}(\mu)) = \boldsymbol{\alpha}_j^T\B F^d(\boldsymbol{x})  + \alpha_{j,0} \;$, but now the optimization is performed with respect to parameters across all layers. There has to be made the same restrictions on $g(\cdot)$ to be continuous and differentiable in the regions of interest as in Strategy~2 if one wants to use gradient based optimizers. One drawback of this strategy is that now there is no guarantee to find a unique global optimum of the likelihood of the feature, even if $g(\cdot)$ is concave. If gradient free optimizers are used the problem becomes computationally extremely demanding given the difficulty of the optimization problem. Furthermore different local optima  define different features having structurally the same configuration and hence the topology of the feature space is getting more complex. 

\item[Strategy 4 (fully Bayesian):] 
All parameters across all layers are drawn from a prior distributions (in the implementation we used $N(0,1)$). There are no restrictions on the nonlinear transformations, only the link function needs to be differentiable. The problem with this strategy is that to get the posterior mode a rather high-dimensional integral has to be solved. The probability of getting close to the mode is extremely low and the convergence requires typically a huge number of iterations.  This might be improved by drawing around the modes obtained by the previously suggested strategies, but to develop this idea is a topic of further research. Just like the 3rd strategy, all different values of the vector of parameters will define different features. With this strategy the joint space of configurations and parameters is (at least in principle) systematically explored, which is extremely demanding computationally. 
\end{description}

In Tables~\ref{t11}-\ref{t13} the predictive performance of these strategies is compared for the NEO asteroids classification problem (Example~\ref{Ex:NeoAsteroids}), the breast cancer data (Example~\ref{Ex:BreastCancer}), and the spam data (Example~\ref{Ex:Spam}). Comparing  Table~\ref{t11} with Table~\ref{t1}, Table~\ref{t12} with Table~\ref{t2} and Table~\ref{t13} with Table~\ref{t3}, we see that there is no substantial difference in predictive performance between the strategies  used for specifying weights.

\begin{table}[H]{
\caption{\label{t11} Comparison of performance (ACC, FPR, FNR) of alternative feature engineer strategies (indicated with $\_2,\_3,\_4$ in the table) for Example~\ref{Ex:NeoAsteroids}. For methods with random outcome the median measures (with minimum and maximum in parentheses) are displayed.
The algorithms are sorted according to median power.}
\resizebox{\textwidth}{!}{
\begin{tabular}{llll}%
\hline 
Algorithm&ACC&FNR&FPR\\\hline
DBRM{\_}G{\_}3&0.9998 (0.9959,1.0000)&0.0002 (0.0001,0.0056)&0.0002 (0.0000,0.0042)\\
DBRM{\_}R{\_}3&0.9998 (0.9953,1.0000)&0.0002 (0.0001,0.0068)&0.0002 (0.0000,0.0070)\\
DBRM{\_}R{\_}4&0.9998 (0.9945,1.0000)&0.0002 (0.0001,0.0080)&0.0002 (0.0000,0.0069)\\
DBRM{\_}G{\_}2&0.9998 (0.9933,1.0000)&0.0002 (0.0001,0.0089)&0.0002 (0.0000,0.0048)\\
DBRM{\_}G{\_}4&0.9998 (0.9932,0.9999)&0.0002 (0.0001,0.0097)&0.0002 (0.0000,0.0042)\\
DBRM{\_}R{\_}2&0.9998 (0.9925,1.0000)&0.0002 (0.0001,0.0105)&0.0002 (0.0000,0.0032)\\
\hline
\end{tabular}}
}
\end{table}

\begin{table}[H]{
\caption{\label{t12}Comparison of performance  (ACC, FPR, FNR) of alternative feature engineer strategies   for Example~\ref{Ex:BreastCancer}. See caption of Table~\ref{t11} for details.}
\resizebox{\textwidth}{!}{
\begin{tabular}{llll}%
\hline 
Algorithm&ACC&FNR&FPR\\\hline
DBRM{\_}G{\_}3&0.9695 (0.9507,0.9789)&0.0536 (0.0479,0.0862)&0.0148 (0.0000,0.0361)\\
DBRM{\_}R{\_}2&0.9695 (0.9554,0.9789)&0.0536 (0.0422,0.0756)&0.0148 (0.0000,0.0396)\\
DBRM{\_}R{\_}4&0.9695 (0.9577,0.9789)&0.0536 (0.0479,0.0756)&0.0148 (0.0000,0.0361)\\
DBRM{\_}R{\_}3&0.9671 (0.9577,0.9789)&0.0536 (0.0422,0.0756)&0.0148 (0.0037,0.0361)\\
DBRM{\_}G{\_}4&0.9671 (0.9577,0.9789)&0.0536 (0.0305,0.0756)&0.0184 (0.0000,0.0361)\\
DBRM{\_}G{\_}2&0.9671 (0.9531,0.9789)&0.0536 (0.0422,0.0862)&0.0184 (0.0000,0.0361)\\
\hline
\end{tabular}}
}
\end{table}

\begin{table}[H]{
\caption{Comparison of performance (ACC, FPR, FNR) of alternative feature engineer strategies   for Example~\ref{Ex:Spam}. See caption of Table~\ref{t11} for details.}\label{t13}
\resizebox{\textwidth}{!}{
\begin{tabular}{llll}%
\hline 
Algorithm&ACC&FNR&FPR\\\hline
DBRM{\_}G{\_}2&0.9243 (0.9100,0.9357)&0.0927 (0.0780,0.1103)&0.0545 (0.0445,0.0686)\\
DBRM{\_}G{\_}3&0.9237 (0.9100,0.9321)&0.0924 (0.0766,0.1122)&0.0548 (0.0474,0.0714)\\
DBRM{\_}G{\_}4&0.9237 (0.9113,0.9315)&0.0931 (0.0821,0.1077)&0.0562 (0.0470,0.0714)\\
DBRM{\_}R{\_}3&0.9240 (0.9132,0.9334)&0.0951 (0.0752,0.1155)&0.0552 (0.0465,0.0672)\\
DBRM{\_}R{\_}2&0.9240 (0.9132,0.9321)&0.0917 (0.0801,0.1142)&0.0550 (0.0465,0.0676)\\
DBRM{\_}R{\_}4&0.9237 (0.9109,0.9341)&0.0931 (0.0787,0.1096)&0.0562 (0.0455,0.0686)\\
\hline
\end{tabular}}
}
\end{table}

\section{Interpretability of DBRM results} \label{ap:interp}

The key feature of DBRM which allows to obtain interpretable models is that there is a whole set $\mathcal{G}$ of non-linear transformations and hence feature engineering becomes highly flexible. To illustrate the importance of the choice of $\mathcal{G}$ we have reanalyzed Example~\ref{Ex:Kepler} on Kepler's third law with DBRM{\_}G{\_}1{\_}PAR{\_}64 using only the sigmoid function as non-linear transformation and considering different restrictions on the search space:
\begin{enumerate}
\item  $\mathcal{G}=\{\text{sigmoid}(x)\}$, $D_{max} = 5$;
\item  $\mathcal{G}=\{\text{sigmoid}(x)\}$, $D_{max} = 300$, and $P_c=0$;
\item  $\mathcal{G}=\{\text{sigmoid}(x)\}$, $D_{max} = 300$, and $P_c=0$ and $p(\gamma_j)\propto 1$.
\end{enumerate}
Clearly for these settings it is not possible to obtain the correct model in closed form, but according to the universal approximation theorem \citep{hurnik1991} Kepler's 3rd law can still be well approximated.
 In the first setting the true model is infeasible since the cubic root function is not a part of $\mathcal{G}$ but at least multiplication of features via the crossover operator is still possible. In the second setting crossovers are not allowed but on the other hand there is no longer any real restriction on the  depth of features. Finally in the third setting all features get a uniform prior in the feature space. As a consequence from this lack of regularization we expect that highly complex features are generated. 
 
 Table~\ref{dnn1} illustrates the effects of making these changes in the DBRM setting on the interpretability of models by reporting the ten most frequently detected features over $N = 100$ simulations. To simplify the reporting we denote \textit{TypeFlag}, \textit{RadiusJpt}, \textit{PeriodDays}, \textit{PlanetaryMassJpt}, \textit{Eccentricity}, \textit{HostStarMassSlrMass}, \textit{HostStarRadiusSlrRad}, \textit{HostStarMetallicity}, \textit{HostStarTempK}, \textit{PlanetaryDensJpt} as $x_1$-$x_{10}$, correspondingly, and use the symbol $\sigma$ for the sigmoid function.  
\begin{table}[H]
\centering
\caption{10 most frequent features detected under Settings 1, 2 and 3} 
\label{dnn1}
\tiny
\begin{tabular}{cl|cl|cl}
\multicolumn{2}{c|}{Setting 1}&\multicolumn{2}{c|}{Setting 2}&\multicolumn{2}{c}{Setting 3}\\\hline
Fq&Feature&Fq&Feature&Fq&Feature\\\hline
99&$x_3$&100&$x_3$&100&$x_3$
\\
98&$x_3$*$x_3$&72&$\sigma$(-10.33+0.24$x_4$-8.83$x_8$)&54&$x_2$
\\
93&$x_3$*$x_{10}$&64&$x_{10}$&21&$\sigma$(-16.91-4.94$x_2$)
\\
4&$x_3$*$x_3$*$x_{10}$&62&$x_2$&19&$x_9$
\\
1&$x_9$*$x_3$&16&$\sigma$(0.21+0.01$x_3$+0.20$x_7$)&16&$x_5$
\\
1&$x_9$*$x_3$*$x_3$&9&$x_4$&14&$x_{10}$
\\
1&$x_{10}$*$x_{10}$*$x_3$&7&$\sigma$(-13.11-7.76$x_8$-3.33$x_2$+0.40$x_{10}$)&10&$\sigma$(6.88$\times10^9$-3.92$x_2$+\\&&&&&3.44$\times10^9$$\sigma$(-13.57-0.17$x_4$-\\&&&&&2.84$x_2$-7.66$x_8$+0.54$x_{10}$)\\&&&&&-13.76$\times10^9$$\sigma$($\sigma$(-13.57-\\&&&&&0.17$x_4$-2.84$x_2$-7.66$x_8$+\\&&&&&0.54$x_{10}$)))
\\
1&$x_7$*$x_3$*$x_3$&5&$\sigma$(-3.36+2.83$x_3$+0.21$x_3$-3.36$x_9$)&9&$x_4$
\\
1&$x_6$*$x_3$*$x_3$&3&$\sigma$($\sigma$(-10.33+0.24$x_4$)-8.83$x_8$)&8&$\sigma$(-13.57-0.17$x_4$-\\&&&&&2.84$x_2$-7.66$x_8$+0.54$x_{10}$)
\\
1&$x_3$*$x_3$*$x_3$&3&$\sigma$(0.15+0.05$x_4$-0.01$x_3$+0.15$x_7$)&7&$\sigma$(0.21+0.21$x_3$)
\\
0&Others&4&Others&$>300$&Others\\
\hline
\end{tabular}
\end{table}

The results shown in Table~\ref{dnn1} are not too surprising. Restricting the set of non-linear transformations results in increasingly more complex features. Using Setting 1 there is not a single occurrence of a sigmoid function while in Setting 2 the feature  $\sigma$(-10.33+0.24$x_4$-8.83$x_8$) is selected in almost 3 out of 4 runs. Removing the complexity penalty in Setting 3 yields highly complex features which are however no longer that much replicable over simulation runs. 

The general conclusion is that more flexible sets of non-linear transformations $\mathcal{G}$ provide the possibility to obtain interpretable models which have similar predictive performance than complex models based on a less flexible set of transformations. Problems with the latter approach include potential overfitting, substantially more need of memory and computational effort (if one for instance is interested in predictions). In contrast DBRM will often construct architectures that reach state of the art performance in terms of prediction and still remain relatively simple, hence representing sophisticated phenomena in a fairly parsimonious way. 

\section{Further applications}
\subsection{Example 6: Simulated data with complex combinatorial structures}\label{ap:sim.ex}

In this simulation study we generated $N = 100$ datasets  with $n=1000$ observations and $p=50$ binary covariates. The covariates were assumed to be independent and were simulated for each simulation run  as $X_{j}\sim \text{Bernoulli}(0.5)$ for $j \in = 1,\dots,50$. In the first simulation study the responses were simulated according to a Gaussian distribution with error variance $\sigma^2 = 1$ and individual expectations specified as follows: 
\begin{eqnarray*}
E(Y) & = & 1 + 1.5 X_{7} + 1.5 X_{8}+ 6.6 X_{18}*X_{21} + 3.5 X_{2}*X_{9}
+ 9 X_{12}* X_{20}* X_{37} 
\\ &+& 7 X_{1}* X_{3}* X_{27} 
+ 7 X_{4}* X_{10}* X_{17}* X_{30} 
+ 7 X_{11}* X_{13}* X_{19}* X_{50}
\end{eqnarray*}
We compare the results of GMJMCMC, RGMJMCMC for DBRM with the Bayesian logic regression model in \citet{Hubin2017}. The latter model differs from the current one in that the model prior is different. For a given logical tree (which is the only allowed feature form) we use $a^{c(L_j)} = \frac{1}{N(s_j)}  \; , \quad  s_j\leq C_{max}$, where $N(s_j)=\binom{m}{s_j}\ 2^{2s_j-2}$. $Q$ and priors for the model parameters are the same as defined in DBRM model. All algorithms were run on 32 threads until the same number of models were visited after the last change of the model space. In particular, in each of the threads the algorithms were run until 20000 unique models were obtained after the last population of models had been generated at iteration 15000. Specification of the Bayesian Logic Regression model corresponds exactly to the one used in simulation Scenario 6 in \citet{Hubin2017}. In this example a detected feature is only counted as a true positive if it exactly coincides with a feature of the data generating model. The results are summarized in Table~\ref{simres}.  Detection in this example corresponds to the features having marginal inclusion probabilities above $\eta^*  = 0.5$ after the search is completed.

Both GMJMCMC and RGMJMCMC performed exceptionally well for fitting this DBRM with slight advantages of the former. %On the other hand the non-adaptive version of RGMJMCMC(NA) only had large power to detect single covariates and two-way interactions but completely failed to detect higher order interactions. Note that as discussed in \citep{Hubin2017} in this kind of analysis the number of false positives often increases with lower power because instead of detecting a correct feature the algorithm might instead detect one or more features closely related to the correct feature.   \todo{Just give one sentence which explains why the NA version does not perform so well} %Similar properties of simple MCMC algorithms (FBLR and MCLR) in such ultra high dimensional scenarios were also shown in \citet{Hubin2017}. 
The original GMJMCMC(LR) algorithm for fitting Bayesian Logic Regression in this case performed almost as well as GMJMCMC and RGMJMCMC, except for a significant drop in power in one of the four-way interactions.  This is however not too surprising because the crossover operator of DBRM models perfectly fits the data generating model whereas the logic regression model focuses on general logic expressions and provides in that sense a larger chance to generate features which are closely related to the data generating four-way interaction \citep{Hubin2017}.
%As discussed in \citet{Hubin2017} standard MCMC algorithms in such ultra high dimensional cases often get stuck in local extrema and mix extremely slowly. Hence the adaptive algorithms seem to be preferable.

\begin{table}[H]
\centering
\caption{Results for Example 6. Power for individual trees, overall power (average power over trees),  expected number of false positives (FP), and false discovery rate (FDR) are compared between GMJMCMC, RGMJMCMC and Bayesian Logic regression.} 
\label{simres}
\begin{tabular}{lccc}%c}
\hline
&DBRM{\_}G&DBRM{\_}R&Bayesian Logic regression \\\hline%&RGMJ(NA)&GMJ(LR) \\\hline
$X_7$&1.0000&1.0000&0.9900\\% & 1.0000&0.9900\\
$X_8$&1.0000&1.0000&1.0000\\%&1.0000&1.0000\\
$X_2*X_9$&1.0000&0.9600&1.0000\\%&0.6800&1.0000\\
$X_{18}*X_{21}$&1.0000&1.0000&0.9600\\%&0.9700&0.9600\\
$X_{1}* X_{3}* X_{27}$&1.0000&1.0000&1.0000\\%&0.0000&1.0000\\
$X_{12}* X_{20}* X_{37}$&1.0000&1.0000&0.9900\\%&0.0000&0.9900\\
$X_{4}* X_{10}* X_{17}* X_{30}$&0.9900&0.9200&0.9100\\%&0.0000&0.9100\\
$X_{11}* X_{13}* X_{19}* X_{50}$&0.9800&0.8900&0.3800\\%&0.0000&0.3800\\
Overall Power&0.9963&0.9712&0.9038\\%&0.4213&0.9038\\
FP&0.5100&1.1400&1.0900\\ %& 14.640&1.0900\\
FDR&0.0601&0.1279&0.1310\\%&0.8453&0.1310\\
\hline\\
\end{tabular}
\end{table}

\subsection{Example 7: Epigenetic data with latent Gaussian variables}\label{ap:ex.epi}

This example illustrates how the extended DBRM model~\eqref{DeepModel2} can be used for feature engineering while simultaneously modeling correlation structures  with latent Gaussian variables. To this end we consider genomic and epigenomic data from \textit{Arabidopsis thaliana}.
Arabidopsis is an extremely well studied model organism for which plenty of genomic and epigenomic data sets are publicly available~\citep[see for example][]{becker2011spontaneous}. 
In terms of epigenetic data we consider methylation markers. DNA locations with a nucleotide of type cytosine nucleobase (C) can be either methylated or not. Our focus will be  on  modeling the amount of methylated reads through different covariates including (local) genomic structures, gene classes and expression levels. The studied data was obtained from the NCBI GEO archive \citep{barrett2013ncbi}, where we consider a sample of $n = 500$ base-pairs chosen from a random genetic region of a single plant. Only cytosine nucleobases can be methylated, hence these 500 observations correspond to 500 sequential cytosine nucleobases from the selected genetic region. 

At each location $i$ there are $R_i$ reads of which $Y_i$ are methylated. Although a binomial distribution would be most appropriate here, we have, due to numerical considerations, assumed a Poisson distribution for $Y_i$ with mean $\mu_i \in \mathbb{R}^{+}$.
%The number of methylated reads $Y_i, \in = 1,...,R_i$ per locus $i = 1,...,n$  is modeled by a Poisson distribution with mean $\mu_i \in \mathbb{R}^{+}$. 
In the extended DBRM model \eqref{DeepModel2} we use the logarithm as the canonic link function. For the feature engineering part of the model we consider $p=14$ input variables which are defined as follows. 
The first factor with three levels is coded with two dummy variables $X_1$ and $X_2$ and describes whether a location belongs to a CGH, CHH or CHG genetic region, where H is either A, C or T. The second factor is concerned with the distance of the location to the previous cytosine nucleobase (C). The dummy variables  $X_3-X_8$ code whether the distance  is  2, 3, 4, 5, from 6 to 20 or greater than 20, respectively,  taking a distance of 1 as reference. 
The third factor describes whether a location belongs to a gene, and if yes whether this gene belongs to a particular group of biological interest. These groups are denoted by $M_\alpha$, $M_\gamma$, $M_\delta$ and $M_0$ and are coded by 3 additional dummy variables $X_{9}-X_{11}$.  Two further covariates are derived from the expression level  for a nucleobase  being either greater than 3000 FPKM or greater than 10000 FPKM, defining binary covariates $X_{12}$ and $X_{13}$. 
The last covariate $X_{14}$ is the offset defined by the total number of reads per location $R_t, \in \mathbb{N}$. The offset mentioned above is modeled as an additional component of the model and hence can be a matter of model choice.

Furthermore we consider the following latent Gaussian variables to model spatial correlations, where marginal likelihoods are computed using the INLA package~\citep{INLAWEB} and the parametrization is taken from there as well:

\begin{description}

\item[$\boldsymbol{AR(1)}$ process:] Autoregressive process of order 1 with parameter $\rho \in \mathbb{R}$, namely $ \delta_i =  \rho\delta_{i-1} + \epsilon_i\in \mathbb{R}$ with $\epsilon_{i} \sim N(0,\sigma_{\epsilon}^2)$, $i = 1,...,n$. For this process the priors on the hyper-parameters are defined as follows: First reparametrize to $\psi_1 = \log{\frac{1}{\sigma_{\epsilon,t}^2}(1-\rho^2)}$, $\psi_2 =   \log{\frac{1+\rho}{1-\rho}}$ \; , \\ then assume
%whilst the following priors have been chosen for the latent Gaussian field:
$\psi_1 \sim \text{logGamma}(1,5\times 10^{-5})$, $\psi_2 \sim N(0,0.15^{-1})$. 

\item[$\boldsymbol{RW(1)}$ process:] Random walk of order 1 based on the Gaussian vector $\delta_1,...,\delta_n$, which is constructed assuming independent increments: $\Delta \delta_i = \delta_i - \delta_{i-1}\sim N(0,\tau^{-1})$. Priors on the hyper-parameters are defined as follows: Reparametrize to $\psi = \log\tau$ and assume $\psi \sim \text{logGamma}(1,5\times10^{-5})$. %\citep{INLAWEB}.

\item[$\boldsymbol{OU}$ process:] Ornstein-Uhlenbeck process (with mean zero), which is defined via the stochastic differential equation $d\delta(t) = -\phi\delta(t)dt + \sigma dW(t)$, where $\phi>0$ and $\{W(t)\}$ is the Wiener process.  This is the continuous time analogue to the discrete time $AR(1)$ model and the process is Markovian. Let $\delta_1,...,\delta_n$ be the values of the process at increasing locations
 $t_1,...,t_n$, then the conditional distribution $\delta_i|\delta_1,...,\delta_{i-1}$ is Gaussian with mean $\delta_{i-1}e^{-\phi z_i}$ and precision $\tau(1-e^{-2\phi z_i})^{-1}$ , where $z_i = t_i - t_{i-1}$ and $\tau = 2\phi/\sigma^2$. Priors on the hyper-parameters are defined as follows: We first reparametrize to $\psi_1 = \log{\tau}$, $\psi_2 = \log{\phi}$ and then assume $\psi_1 \sim \text{logGamma}(1,5\times 10^{-5})$, $\psi_2 \sim N(0,0.2^{-1})$. %\citep{INLAWEB}.

\item[$\boldsymbol{IG}$ process:] Independent Gaussian process $\{\delta_i\}$ with $\delta_i\sim N(0,\tau^{-1})$. Priors on the hyper-parameters are defined as follows: First reparametrize to $\psi = \log\tau$ and then assume $\psi \sim \text{logGamma}(1,5\times10^{-5})$. %\citep{INLAWEB}.
\end{description} 
These different processes allow to model different spatial dependence structures of methylation rates along the genome. They can also account for variance  which is not explained by the covariates. DBRM can be used to find the best combination of latent variables for modeling this dependence in combination with deep feature engineering. 

The Bayesian model is completed with Gaussian priors for the regression coefficients
\begin{align}
\boldsymbol{\beta}|\boldsymbol{\gamma} \sim & N_{p_{\boldsymbol{\gamma}}}(\boldsymbol{0},I_{p_\gamma}e^{-\psi_{\beta_\gamma}})\\
\psi_{\beta_\gamma}\sim & \text{logGamma}(1,5\times10^{-5})
\end{align}
%~\eqref{glmbetarprior} 

We then use prior~\eqref{eq:modelprior} with $a=e^{-2 \log n}$ for $\boldsymbol{\gamma}$. We use a similar prior for $\B \lambda$ associated with selection of the latent Gaussian variables, 
\begin{align}
  p(\boldsymbol\lambda)\propto &  \prod_{j=1}^{r}\exp(-2\log n\lambda_j)\; ,\label{lammodel1}
\end{align}
 where each latent Gaussian variables has equal prior probability to be included.

\begin{center}

\begin{table}[tb]
\centering
\caption{\label{tepi} Results for Example 7: Features and latent Gaussian variables (LGV) with posterior probability above 0.25 found by GMJMCMC using 16 parallel threads.}
\begin{tabular}{l|lc}%
\hline 
&Variable&Posterior\\
\hline 
Features&offset(log(total.bases))&1\\
&CG&0.999\\
&CHG&0.952\\
\hline
LGV&RW(1)&1\\
\hline
\end{tabular}
\end{table}

\end{center}

From the results of Table~\ref{tepi} we learn that there are three features with large posterior probability: the offset for the total number of observations per location as well as two features indicating whether the location is CG or CHG. Among the latent Gaussian variables only the random walk process of order one was found to be of importance. None of the engineered features were found of importance for this example. Like in Example 1 and 2 we observe that although our feature space includes highly non-linear features the regularization due to our priors guarantees the choice of parsimonious models and non-linear features are only selected if really necessary.

\end{document}